\documentclass[11pt,epsf]{article}

\long\def\comment#1{\ifdim\overfullrule>0pt{\sf[{#1}]}\fi}

\usepackage{latexsym}
\usepackage{amsmath,amsfonts}
\usepackage{epsfig}
\usepackage{bbm}
\usepackage{amsthm,amssymb,hyperref}
\usepackage{boxedminipage}
\usepackage{mathtools}
\usepackage{subcaption}
\usepackage{enumitem}
\usepackage{mdframed}
\usepackage{xspace}
\usepackage{mathrsfs}
\usepackage{tikz}
\usepackage{pgfplots}
\usepackage{pgfplotstable}
\usepackage{tikzscale}
\usepackage[ruled]{algorithm2e}
\usepackage{complexity}
\usepackage{thm-restate}

\usepackage{soul}
\usepackage{todonotes}

\newcommand{\Hd}{{\cal H}}
\newcommand{\V}{\cal V}
\newcommand{\Ed}{\cal E}
\newcommand{\cch}{C_{\Hd}}

\usetikzlibrary{decorations.markings}
\usetikzlibrary{shapes}
\usetikzlibrary{arrows}

\newtheorem{theorem}{Theorem}
\newtheorem{lemma}[theorem]{Lemma}

\newtheorem{corollary}[theorem]{Corollary}

\newtheorem{observation}[theorem]{Observation}

\newtheorem{claim}[theorem]{Claim}
\newtheorem{definition}[theorem]{Definition}

\newtheorem{problem}{Problem}

\numberwithin{equation}{section}
\numberwithin{theorem}{section}
\numberwithin{lemma}{section}
\numberwithin{conjecture}{section}
\numberwithin{corollary}{section}
\numberwithin{proposition}{section}
\numberwithin{observation}{section}
\numberwithin{claim}{section}
\numberwithin{definition}{section}
\numberwithin{question}{section}
\numberwithin{fact}{section}

\newcommand{\npm}{N^{\pm}}
\newcommand{\ndecomp}{path decomposition\xspace}

\newcommand{\ab}{*}

\newcommand{\lightCol}{8}

\newenvironment{cproof}
{\begin{proof}
 [Proof.]
 \vspace{-1.5\parsep}
}
{ \end{proof}}

\begin{document} 
\bibliographystyle{alpha}
\def\proofend{\hfill$\Box$\medskip}
\def\Proof{\noindent{\bf Proof:\ \ }}
\def\Sketch{\noindent{\bf Sketch:\ \ }}
\def\eps{\epsilon}

\title{Coloring tournaments with few colors: Algorithms and complexity}
\author{Felix Klingelhoefer$^*$ \and Alantha Newman\thanks{Laboratoire
    G-SCOP (CNRS, Univ. Grenoble Alpes), Grenoble, France.  Supported
    in part by ANR project DAGDigDec (ANR-21-CE48-0012).   
    {\tt{firstname.lastname@grenoble-inp.fr}}
    \newline \newline
 A preliminary version of these results appeared in the Proceedings of the 31st Annual European Symposium on Algorithms (ESA 2023).
}}

\date{}

\maketitle

\begin{abstract}

A \emph{$k$-coloring} of a tournament is a partition of its vertices
into $k$ acyclic sets.  Deciding if a tournament is 2-colorable is
\NP-hard.  A natural problem, akin to that of coloring a 3-colorable
graph with few colors, is to color a 2-colorable tournament with few
colors.  This problem does not seem to have been addressed before,
although it is a special case of coloring a 2-colorable 3-uniform
hypergraph with few colors, which is a well-studied problem with
super-constant lower bounds.

We present a new efficient decomposition lemma for tournaments, which
we use to design polynomial-time algorithms to color various
classes of tournaments with few colors, notably, to color
a 2-colorable tournament with ten colors.  We also use this lemma to
prove equivalence between the problems of coloring 3-colorable
tournaments and coloring 3-colorable graphs with constantly many
colors.  For the classes of tournaments considered, we complement our
upper bounds with strengthened lower bounds, painting a comprehensive
picture of the algorithmic and complexity aspects of coloring
tournaments.
\end{abstract}

\section{Introduction}

A tournament $T=(V,A)$ is a complete, oriented graph: For each pair of
vertices $i,j \in V$, there is either an arc from $i$ to $j$ or an arc
from $j$ to $i$ (but not both).  A subset of vertices $S \subseteq V$
induces the {\em subtournament} $T[S]$.  If this subtournament
contains no directed cycles, then it is said to be {\em acyclic}.  The
problem of {\em coloring a tournament} is that of partitioning the
vertices into the minimum number of acyclic sets, sometimes referred
to as the {\em dichromatic number}~\cite{neumann1982dichromatic}.
Since a tournament contains a directed cycle if and only if it
contains a directed triangle, the problem of coloring a tournament is
equivalent to partitioning the vertices into the minimum number of
sets so that each set does not contain a directed triangle.
Conversely, it is equivalent to problem of coloring the vertices with
the minimum number of colors so that each directed triangle has at
least two colors.

Coloring tournaments can be compared to the problem of coloring
undirected graphs.  For the latter, deciding if a graph is 2-colorable
(i.e., bipartite) is easy, but it is \NP-hard to decide if a graph is
3-colorable.  A problem (e.g., deciding if a graph $G$ is bipartite)
is ``easy'', if for an (unweighted) graph $G$ on $n$ vertices, there
is an algorithm that solves this problem on $G$ and this algorithm
runs in time polynomial in $n$.  More generally, we say that a
procedure runs in {\em polynomial time} or is a {\em polynomial-time
  algorithm} if its running time is polynomial in the size of the
input.  When a problem has such a polynomial-time algorithm, we say it
can be solved {\em efficiently}.  A problem that is \NP-hard is
unlikely to have a polynomial-time algorithm.

In a widely-studied promise problem, we are given a graph promised to
be 3-colorable and the goal is to color it (in polynomial time) with
few
colors~\cite{wigderson1983improving,blum1994new,KargerMS98,kawarabayashi2017coloring}.
For tournaments, it is easy to decide whether or not a tournament is
1-colorable (i.e., transitive), since this is exactly when the
tournament is acyclic.  However, deciding if a tournament is
2-colorable is already \NP-hard, as shown by \cite{chen2007min} in
response to a question of Andr\'as Frank asking about the complexity
of deciding if the vertex set of a tournament can be partitioned into
two feedback vertex sets.

This suggests the following promise problem: Given a tournament
promised to be 2-colorable, what is the fewest number of colors with
which it can be colored in polynomial time?  This question is the
starting point for this paper and naturally leads to related problems
of determining upper and lower bounds for coloring various classes of
tournaments.  For comparison, the complexity landscape of graph
coloring is well studied and we have a general understanding of what
it looks like. (See Table \ref{tab:gen_graph}.)  In contrast, the
problem of coloring tournaments has been studied very little from the
algorithmic or complexity perspective.  It has however been studied
extensively from the perspective of structural graph theory (e.g.,
\cite{berger2013tournaments,harutyunyan2019locToGlobal,nguyen2023some,chudnovsky2024pure}). But
beyond some basic \NP-completeness results~\cite{chen2007min}, the
most obvious complexity questions remained open.  This paper is an
effort to address this disparity.

\begin{table}
\begin{small}
        \begin{center}
           \begin{tabular}{| l | c | c | c|}
             \hline
                   \textbf{Graph Type} & \textbf{Lower Bound} &
                   \textbf{Upper Bound} \\
                   \hline &  & \\ 
                   $3$-Colorable graphs & $5$
                   \cite{bulin2019algebraic}, $O(1)^{\ab}$
                   \cite{guruswami2020d} &   $\widetilde{O}(n^{0.19996})$\cite{kawarabayashi2017coloring} \\[1ex]
                   \hline & & \\
		   $k$-Colorable graphs, $k\geq 3$ & $2 k -1$
                   \cite{bulin2019algebraic}, $O(1)^{\ab}$\cite{guruswami2020d} &  $O(n^{1-\frac{3}{k+1}})$ \cite{DBLP:journals/jacm/KargerMS98} \\[1ex]
                   \hline & & \\
                   General graphs    & $n^{1-\epsilon}$
                   \cite{hastadClique,zuckerman2006linear}  & $O(n(\log \log n)^2(\log n)^{-3})$ $^{\dagger}$\cite{halldorsson1993still} \\[1ex]
                   \hline & & \\
                   3-Uniform 2-colorable hypergraphs &  $O(1)$
                   \cite{dinur2005hardness}  & $\widetilde{O}(n^{\frac{1}{5}})$ \cite{krivelevich2001approximating} \\[1ex] \hline
           \end{tabular}
         \end{center}
\end{small}
\caption{Best-known lower and upper bounds for various graph coloring
  problems.  All
  inapproximability results are under the assumption $\P \neq \NP$
  except those denoted by $^{\ab}$, which are under the $d$-To-1
  Conjecture~\cite{khot2002power}. The lower bound should be read as,
  ``It is hard to color a 3-colorable graph with 5 colors.'' The upper
  bound as, ``A 3-colorable graph can be (efficiently) colored with $\widetilde{O}(n^{0.19996})$
  colors.''  The exception is the entry
  indicated by $^\dagger$, which is a
  hardness of approximation result.  }\label{tab:gen_graph}
\end{table}

\paragraph{Previous Work.}

The problem of coloring a 2-colorable tournament with few colors is a
special case of coloring a 2-colorable 3-uniform hypergraph with few
colors.\footnote{This follows from the aforementioned observation that
  the problem of coloring a tournament is equivalent to finding a
  coloring such that each triangle has at least two colors.}  Deciding
if a 3-uniform hypergraph is 2-colorable is
\NP-hard~\cite{lovasz1973coverings} and more recently it was proved to
be \NP-hard to color with any constant number of
colors~\cite{dinur2005hardness}.  On the positive side, a 2-colorable
3-uniform hypergraph can be colored in polynomial time with
$\widetilde{O}(n^{1/5})$ colors~\cite{alon1996coloring,
  chen2005coloring,krivelevich2001approximating}, a result which uses
tools from and is analogous to that of
\cite{KargerMS98} for 3-colorable
graphs.\footnote{The (standard) notation $\widetilde{O}$ ignores
  logarithmic factors.}  Thus, $\widetilde{O}(n^{1/5})$ is an upper
bound on the number of colors needed to efficiently color a
2-colorable tournament and was the best-known upper bound prior to our
work.  Deciding if a tournament is 2-colorable is
\NP-hard~\cite{chen2007min} and furthermore, deciding if a tournament
is $k$-colorable for any $k \geq 2$ is \NP-hard~\cite{fox2019removal}.
These results do not rule out the existence of a polynomial-time
algorithm to color a 2-colorable tournament with three colors.

\begin{table}
\begin{small}
        \begin{center}
           \begin{tabular}{| l | c | c | }
             \hline
             \textbf{Tournament Type} & \textbf{Lower Bound} &
             \textbf{Upper Bound} \\ \hline &  & \\
             2-Colorable tournaments & 2\cite{chen2007min}, 3  & 10
             \\[1ex] \hline &  & \\
		   3-Colorable tournaments & 5, $O(1)$ $^{\ab}$ & $\widetilde{O}(n^{0.19996})$ \\[1ex]\hline & & \\
		   $k$-Colorable tournaments, $k\geq 2$ & $2k - 1$,
		   $O(1)$ $^{\ab}$ &  $5\cdot f(k-1) \cdot g(k)$
                   \\[1ex] \hline & & \\
                   2-Colorable light tournaments & in \P ? & 5
                   \\[1ex] \hline & & \\
                   Light tournaments & in \P ?
                   & \lightCol
                   \\[1ex] \hline & & \\
		   General tournaments & $n^{\frac{1}{2}-\epsilon}$ $^{\dagger}$
                   & $n/\log{n}$\cite{erdos1964representation} \\[1ex] \hline
           \end{tabular}
         \end{center}
\end{small}
\caption{Best-known lower and upper bounds for various tournament
  coloring problems.  Previous results are indicated with a citation.
  All the results without a citation are established in this paper.
  Lower bounds are under the assumption $\P \neq \NP$ except those
  marked with a $^{\ab}$, which hold under the $d$-To-1
  Conjecture~\cite{khot2002power}.  The function $g(k)$ denotes the
  number of colors needed to efficiently color a $k$-colorable graph,
  while $f(k)$ is the number of colors needed to efficiently color a
  $k$-colorable tournament.  As in Table \ref{tab:gen_graph}, the
  lower bound should be read as, ``It is hard to color a 2-colorable
  tournament with 3 colors.''  The upper bound as, ``A 2-colorable
  tournament can be (efficiently) colored with 10 colors.''  The
  exception is the entry indicated by $^\dagger$, which is a hardness
  of approximation result.  }\label{tab:gen_tourn}
\end{table}

From the perspective of structural graph theory, the problem of
coloring tournaments has been widely studied due to its connection to
the famous Erd\H{o}s-Hajnal
Conjecture~\cite{erdos1989ramsey,chudnovsky2014erdos}, which has an
equivalent formulation in terms of tournaments~\cite{alon2001ramsey}.
The latter posits that for any tournament $H$, there is a constant
$\eps_H$ (where $0 < \eps_H \leq 1$) such that any $H$-free tournament
on $n$ vertices has a transitive subtournament of size at least
$O(n^{\eps_H})$.  Tournaments $H$ for which $\eps_H = 1$
  are called {\em heroes} and have been characterized by
  \cite{berger2013tournaments}.  Forbidding any hero $H$ in a
tournament $T$ actually implies that $T$ has constant dichromatic
number~\cite{berger2013tournaments}, which yields a transitive induced
subtournament of linear size.  These results are existential and do
not provide an efficient algorithm to color an $H$-free tournament
with a constant number of colors when $H$ is some fixed hero.

\paragraph{Our Results.}
We consider some basic algorithmic and computational complexity
questions on the subject of coloring tournaments.  Our main
algorithmic tool, presented in Section \ref{sec:decomp}, is a
decomposition lemma which can be used to obtain efficient algorithms
for coloring tournaments in various cases.  On a high level, it bears
some resemblance to decompositions previously used to prove bounded
dichromatic number in tournaments and in dense digraphs with forbidden
subgraphs~\cite{berger2013tournaments,harutyunyan2019coloring}.  To
apply our decomposition lemma to 2-colorable tournaments, we use an
observation from
\cite{alon1996coloring,chen2005coloring,krivelevich2001approximating}
which states that there is an efficient algorithm to partition a
2-colorable tournament into two tournaments that are each light.  A
{\em light tournament} is one in which for each arc $uv$, the set of
vertices $N(uv) = \{w ~ |~ uvw \text{ forms a directed triangle}\}$ is
transitive.  \cite{berger2013tournaments} showed that light
tournaments have constant dichromatic number, since light tournaments
are exactly those in which a certain hero is forbidden as an induced
subtournament.

Combining these observations, we can do the following for any
tournament: we can either partition it into two light tournaments and
conclude that it can be colored with $O(1)$ colors, or we have a
certificate that it is not 2-colorable.  Thus, coloring a 2-colorable
tournament with $O(1)$ colors cannot be an \NP-hard problem, unless
\NP = \coNP.  This does not however immediately imply that there is an
efficient algorithm, since there are many search problems that are
believed to be intractable even though their decision variant is easy
(e.g., those in the class \TFNP).  We remark that although
\cite{berger2013tournaments} did not provide an efficient algorithm to
color a light tournament with a constant number of colors, a careful
modification of their techniques indeed results in a polynomial-time
algorithm using around 35 colors to color a light tournament, yielding
an efficient algorithm to color a 2-colorable tournament with at most
70 colors.  Details can be found in \cite{FelixThesis}.

In Section \ref{sec:algs}, we give applications of our algorithmic
decomposition lemma to color various classes of tournaments.
Specifically, we show that 2-colorable tournaments can be efficiently
colored with ten colors.  We then use our toolbox to study 3-colorable
tournaments.  Here we show that the problem of coloring a 3-colorable
tournament has a constant-factor reduction to the problem of coloring
3-colorable graphs; specifically, if we can color a 3-colorable graph
with $k$ colors, then we can color a 3-colorable tournament with $50k$
colors.  We also use our tools to show that light tournaments can be
efficiently colored with eight colors, but since this is more
technically involved than the other cases, we defer this to Section
\ref{sec:light2}.

Next, we strengthen the lower bounds by showing in Section
\ref{sec:hardness} that it is \NP-hard to color a 2-colorable
tournament with three colors.  We then give a reduction from coloring
graphs to coloring tournaments, which implies, for example, that it is
hard to color 3-colorable tournaments with $O(1)$ colors under the
$d$-To-1 Conjecture of Khot~\cite{khot2002power}.  Finally, we show
that it is \NP-hard to approximate the number of colors required for a
general tournament to within a factor of $O(n^{1/2-\epsilon})$ for any
$\epsilon > 0$.  Our results are summarized in Table
\ref{tab:gen_tourn}.

We observe that, like another theorem, which shows that the dichromatic
number of a tournament is bounded (i.e., constant) if the
out-neighborhoods of vertices have bounded dichromatic
number~\cite{harutyunyan2019locToGlobal}, the decomposition lemma in
Section \ref{sec:decomp} also has a local-to-global flavor: If the
sets $N(uv)$ can be efficiently colored with few colors for all arcs
$uv$ and if there are two vertices $s$ and $t$ such that the
out-neighborhood of $s$ and the in-neighborhood of $t$ can be
efficiently colored with few colors, then our decomposition lemma
yields an efficient algorithm to color the whole tournament with few
colors.  Thus, this decomposition lemma, which links the dichromatic
number of a tournament and the dichromatic number of the sets $N(uv)$,
or {\em arc neighborhoods}, could be a useful tool for bounding
dichromatic number in other settings and likely has further
applications.  For example, it was recently used as a key
tool~\cite{klingelhoefer2023bounding} in establishing the equivalence
between a conjecture of \cite{nguyen2023some} and a conjecture of
\cite{el1985existence}.

\subsection{Notation and Preliminaries}\label{sec:notation}

Let $T=(V,A)$ be a tournament with vertex set $V$ and arc set $A$.
Sometimes, we use $V(T)$ to denote its vertex set and $A(T)$ to denote
its arc set.  For $S \subset V$, we use $T[S]$ to denote the
subtournament induced on vertex set $S$, although we sometimes abuse
notation and refer to the subtournament itself as $S$.  We define $uv
\in A$ to be an arc directed from $u$ to $v$.  We define $N^+(v)$ to
be all $w \in V$ such that arc $vw \in A$ and $N^-(v)$ to be all $w
\in V$ such that arc $wv \in A$.  We let $N^+[v] = N^+(v)\cup \{v\}$
and $N^-[v] = N^-(v) \cup \{v\}$.  For $S \subset V$, we define
$N^+(S) = \bigcup_{v \in S} N^+(v)$, and we define $N^-(S), N^+[S]$
and $N^-[S]$ analogously.  We use $\npm(S)$ to denote vertices in
$V\setminus{S}$ that have at least one in-neighbor and at least one
out-neighbor in $S$, which we call the {\em mixed neighborhood} of the
set $S$.

For $S, U \subset V$ such that $S \cap U = \emptyset$, we use $S
\Rightarrow U$ to indicate that all arcs between $S$ and $U$ are
directed from $S$ to $U$.  Let $C_3$ denote a directed triangle;
usually, we refer to this simply as a triangle.  Define $N(uv) \subset
V$ to contain all vertices $w$ such that $uvw$ forms the directed
triangle consisting of arcs $uv, vw$ and $wu$.  In other words, $N(uv)
= N^-(u) \cap N^+(v)$.  For three (vertex disjoint) tournaments $T_1,
T_2$ and $T_3$, we use $\Delta(T_1, T_2, T_3)$ to denote the
tournament resulting from adding all arcs from $T_1$ to $T_2$, all
arcs from $T_2$ to $T_3$ and all arcs from $T_3$ to $T_1$.

A tournament $T=(V,A)$ is {\em $k$-colorable} if there is a partition
of $V$ into $k$ vertex-disjoint sets, $V_1, V_2, \ldots, V_k$, such
that $T[V_i]$ is transitive for all $i \in \{1, \ldots, k\}$.  We use
$\vec\chi(T)$ to denote the {\emph{dichromatic number}} of $T$ (i.e.,
the minimum number of transitive subtournaments into which $V(T)$ can
be partitioned).  When the context is clear, we refer to the
dichromatic number simply as the {\em chromatic number}.  As
mentioned, computing the value $\vec\chi(T)$ is in general
\NP-hard~\cite{chen2007min}.  Our goal is to find upper and lower
bounds on the number of colors by which a tournament $T$ can be
colored using an {\em efficient} or {\em polynomial-time} algorithm,
which is an algorithm whose running time is polynomial in the size of
$T$.

We remark that we will always assume that a tournament $T$ which we
want to color is strongly connected; if this were not the case, we can
color each strongly connected component separately.  Therefore, each
vertex has an out-neighborhood containing at least one vertex.

\section{Efficient Tournament Decomposition for Coloring}\label{sec:decomp}

We present a decomposition for certain tournaments that can be
computed in polynomial time and yields an efficient method to color
such tournaments with few colors.

\begin{definition}\label{def:vc_gen}
A sequence of vertices $(v_i)_{0\leq i \leq k}$ in a tournament $T$ is called a 
\emph{vertex chain}, if for a pair of distinct vertices $v_0$ and $v_k$, 
$(v_i)_{0\leq i \leq k}$ are the vertices in a (fixed) shortest directed path
from $v_0$ to $v_k$.
Additionally, we define an \emph{arc chain} $(e_i)_{1 \leq i \leq
  k}$ corresponding to a vertex chain, where $e_i$ is the arc from
$v_{i-1}$ to $v_i$.
\end{definition}

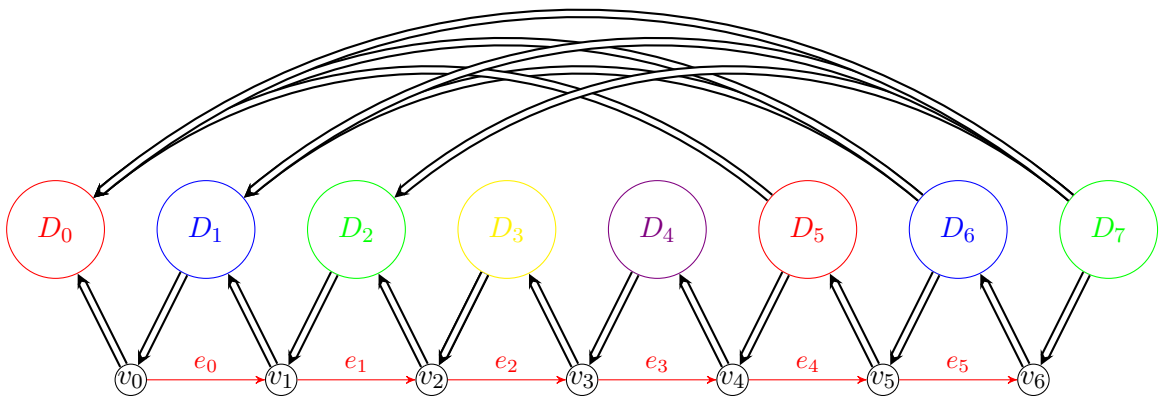
\begin{figure}
        \center
	\begin{tikzpicture}[->,>=stealth']
		  \tikzstyle{smallvertex}=[circle,draw,minimum size=8pt,inner sep=0pt]
		  \tikzstyle{bigvertex}=[circle,draw,minimum size=37pt,inner sep=0pt]
		  \tikzstyle{group}==[circle,draw,minimum size=100pt,inner sep=0pt]

		\node[smallvertex] (v0) at (0,0) {$v_0$};
		\node[smallvertex] (v1) at (2,0) {$v_1$};
		\node[smallvertex] (v2) at (4,0) {$v_2$};
		\node[smallvertex] (v3) at (6,0) {$v_3$};
		\node[smallvertex] (v4) at (8,0) {$v_4$};
		\node[smallvertex] (v5) at (10,0) {$v_5$};
		\node[smallvertex] (v6) at (12,0) {$v_6$};

		\node[bigvertex,red] (D0) at (-1,2) {$D_0$};
		\node[bigvertex,blue] (D1) at (1,2) {$D_1$};
		\node[bigvertex,green] (D2) at (3,2) {$D_2$};
		\node[bigvertex,yellow] (D3) at (5,2) {$D_3$};
		\node[bigvertex,violet] (D4) at (7,2) {$D_4$};
		\node[bigvertex,red] (D5) at (9,2) {$D_5$};
		\node[bigvertex,blue] (D6) at (11,2) {$D_6$};
		\node[bigvertex,green] (D7) at (13,2) {$D_7$};

		\node[text=red] at (1,0.2) {\small $e_1$};
		\node[text=red] at (3,0.2) {\small $e_2$};
		\node[text=red] at (5,0.2) {\small $e_3$};
		\node[text=red] at (7,0.2) {\small $e_4$};
		\node[text=red] at (9,0.2) {\small $e_5$};
		\node[text=red] at (11,0.2) {\small $e_6$};

		\draw [->,>=stealth,thick,double distance=2pt] (v0) to (D0);
		\draw [->,>=stealth,thick,double distance=2pt] (D1) to (v0);
		\draw [->,>=stealth,thick,double distance=2pt] (v1) to (D1);
		\draw [->,>=stealth,thick,double distance=2pt] (D2) to (v1);
		\draw [->,>=stealth,thick,double distance=2pt] (v2) to (D2);
		\draw [->,>=stealth,thick,double distance=2pt] (D3) to (v2);
		\draw [->,>=stealth,thick,double distance=2pt] (D3) to (v2);
		\draw [->,>=stealth,thick,double distance=2pt] (v3) to (D3);
		\draw [->,>=stealth,thick,double distance=2pt] (D4) to (v3);
		\draw [->,>=stealth,thick,double distance=2pt] (v4) to (D4);
		\draw [->,>=stealth,thick,double distance=2pt] (D5) to (v4);
		\draw [->,>=stealth,thick,double distance=2pt] (v5) to (D5);
		\draw [->,>=stealth,thick,double distance=2pt] (D6) to (v5);
		\draw [->,>=stealth,thick,double distance=2pt] (v6) to (D6);
		\draw [->,>=stealth,thick,double distance=2pt] (D7) to (v6);

		\draw [->,>=stealth,thick,double distance=2pt, bend right, out=-40, in=-140] (D6) to (D1);
		\draw [->,>=stealth,thick,double distance=2pt, bend right, out=-40, in=-140] (D6) to (D0);
		\draw [->,>=stealth,thick,double distance=2pt, bend right, out=-40, in=-140] (D5) to (D0);
		\draw [->,>=stealth,thick,double distance=2pt, bend right, out=-40, in=-140] (D7) to (D2);
		\draw [->,>=stealth,thick,double distance=2pt, bend right, out=-40, in=-140] (D7) to (D1);
		\draw [->,>=stealth,thick,double distance=2pt, bend right, out=-40, in=-140] (D7) to (D0);

		    \path[every node/.style={font=\sffamily\small},red]
		        (v1) edge node [right] {} (v2)
		        (v2) edge node [right] {} (v3)
		        (v3) edge node [right] {} (v4)
		        (v4) edge node [right] {} (v5)
		        (v5) edge node [right] {} (v6)
		        (v0) edge node [right] {} (v1);

        \end{tikzpicture}
        \caption{A \ndecomp of $T$. The red arcs $(e_i)$ form a
          shortest path from $v_0$ to $v_k$, thus all the arcs not
          depicted between the $v_i$'s go backwards. All the vertices
          in a given $D_i$ are colored from the color palette
          indicated by the color of the $D_i$.  Notice that because
          there are no long forward arcs between the $D_i$'s, all arcs
          between $D_i$'s that share a color palette are backwards.
	\label{fig:decomp_D}}
\end{figure}

\begin{definition}
A vertex chain $(v_i)_{0 \leq i \leq k}$ is called a $(c,d)$-vertex
chain, if the following properties hold: $T[(N^+(v_0) \cup N^-(v_k))]$
can be colored with $c$ colors in polynomial time, and $T[N(e_i)]$ can
be colored with $d$ colors in polynomial time for all $i \in
\{1,\ldots,k\}$.
\end{definition}

\begin{definition}\label{def:edge} Given a vertex chain $(v_i)_{0 \leq i \leq k}$, 
  a \emph{\ndecomp}
  of a tournament $T$ is defined as:
  	\begin{itemize}
		\item $D_0 = N^+(v_0)$,
		\item For $1 \leq i \leq k$, $D_i = N(e_i) \setminus (\cup_{0\leq j \leq i-1} D_j)$,
		\item $D_{k+1} = N^-(v_k) \setminus (\cup_{0\leq j \leq k} D_j)$.
	\end{itemize}
\end{definition}

The main idea behind this decomposition is to build zones that can be
efficiently colored, such that all arcs between zones at distance more
than four (i.e., \textit{long arcs}) go backwards.  A depiction of a
\ndecomp is shown in Figure \ref{fig:decomp_D}.  We now prove that
this is indeed a decomposition of $T$.

\begin{lemma}\label{lem:decomp}
Let $(D_0, \ldots, D_{k+1})$ be a \ndecomp of a tournament $T$.  Then
$V = \cup_{0 \leq i \leq {k+1}}D_i$.
\end{lemma}

\begin{proof}
We will prove this lemma by contradiction: Suppose there is a vertex
$w \in V$ that does not belong to any $D_i$.  Assume that $w$ does not
belong to the vertex chain.  Since $w$ is neither in $D_0$ nor in
$D_{k+1}$, it follows that $w \in N^-(v_0)$ and $w \in N^+(v_k)$.
Take the smallest integer $i$ such that $w \in N^+(v_i)$.  There must
be one since $w \in N^+(v_k)$.  Notice that $i \geq 1$ since $w \notin
N^+(v_0)$, so $e_i$ belongs to the arc chain and $w \in N(e_i)$.
Therefore, $w \in D_i$, which is a contradiction.

Now consider the case in which $w$ is in the vertex chain.  An arc
with both endpoints in the vertex chain that is not in the arc chain
is backwards.  Thus, $v_i \in N(e_{i+2})$ for all $0 \leq i \leq k-2$.
Notice that $v_{k-1}$ can belong to $D_{k+1}$ (if it does not belong
to $D_j$ for some $j < k+1$).  Finally, $v_k \in N(e_{k-1})$.
\end{proof}

For the sake of simplicity and to more easily visualize the
decomposition, we remark that it might be easier to not include the
vertices in the vertex chain in the \ndecomp.  In this case, these
vertices can be colored with two extra colors.  Since all arcs not in
the arc chain with both endpoints in the vertex chain go backwards
(with respect to the arc chain; otherwise there would be a shorter
path from $v_0$ to $v_k$), we can use two colors so that all forwards
arcs (those in the arc chain) are bicolored.

\begin{lemma}\label{lem:long_arc}
Let $(D_0, \ldots, D_{k+1})$ be a path decomposition of a tournament
$T$.  Let $0 \leq i,j \leq k+1$ and let $j \geq i+5$.  For $u\in D_i$
and $w\in D_j$, we have $u\in N^+(w)$.
\end{lemma}

\begin{proof}
We will prove this by contradiction. Suppose $j \geq i+5$ and $u \in
N^-(w)$.  Then there is a path of three arcs from $v_i$ to $v_{j-1}$,
namely $(v_i, u, w, v_{j-1})$.  (By definition of the decomposition,
$u \in D_i$ implies $u \in N^+(v_i)$ and $w \in D_j$ implies $w \in
N^-(v_{j-1})$.) This is not possible since by the definition of the
vertex chain as a shortest path, there can be no path between $v_i$
and $v_{j-1}$ with fewer than four arcs (since $(j-1)-i \geq (i+5 -1)
-i = 4$).
\end{proof}

\begin{lemma}\label{lem:eff_loc_to_glob}
If $T$ has a $(c,d)$-vertex chain $(v_i)_{0\leq i \leq k}$ that can be
found in polynomial time, where $c \geq d$, then $T$ can be colored with at most $c + 4d$ colors in polynomial time.
\end{lemma}
  
\begin{proof}
  We construct a \ndecomp for the $(c,d)$-vertex chain.
  We first
consider the case in which $k \leq 3$.  We color $D_0 \cup
D_{k+1}$ with $c$ colors.  We then make three palettes of $d$ colors
each with labels from $0$ to $2$, and for $i \in \{1,2,3\}$, we color
each $D_i$ using the color palette with label $i \bmod 3$.  This uses
a total of $c+3d$ colors.  Notice that the only $i$ and $j$ such that
there are arcs $uv$ with $u \in D_i$ and $v \in D_j$ is when $i=0$ and
$j=k+1$.  Thus, all the arcs between different $D_i$ are bicolored.
Notice that we use at most $c + 3d$ colors and that each of the four
sets can be colored efficiently with the respective number of colors.

Now we consider the case in which $k \geq 4$.  We make five palettes
of $d$ colors each with labels from $0$ to $4$.  We color
each $D_i$ using the color palette with label $i \bmod 5$.
For $D_0$ and $D_{k+1}$, we use the same set of $c-d$ ``extra'' colors
(since all arcs go from $D_{k+1}$ to $D_0$).
Let us
show that we can do this in polynomial time.  Notice that $D_0
\subseteq N^+(v_0)$, and $D_{k+1} \subseteq N^+(v_k)$.  By definition
of $(c,d)$-vertex chain, each of these sets can be colored efficiently
with $c$ colors.  For every $1 \leq i \leq k$, $D_i$ is a subset of
$N(e_i)$, which can be colored efficiently with $d$ colors.

We now show that this is a proper coloring of $T$.  We will do this by
showing that all forward arcs between different $D_i$ are bicolored.
For an arc $uv$ with $u \in D_i$ and $v \in D_j$, we say that $uv$ is
a {\emph{forward arc}} if $i < j$.  By Lemma \ref{lem:long_arc}, there
are no forward arcs between $D_i$ and $D_j$ when $j\geq i+5$.
Furthermore, by the definition of the coloring, no vertex in $D_i$ and
$D_j$ can share a color for $i+1 \leq j \leq i+4$.  Thus all forward
arcs from $D_i$ to $D_j$ will be bicolored.  Since every $D_i$ is
properly colored, and all forward arcs between different $D_i$ are
bicolored, $T$ is properly colored.  If $c=d$, the total number of
colors used is $5c$.  If $c > d$, then the coloring uses at most
$(c-d) + 5d = c+4d$ colors.\end{proof}

\section{Algorithms for Coloring Tournaments}\label{sec:algs}

In this section, we consider the problems of coloring 2-colorable and
3-colorable tournaments, and we show how to use our tools to
efficiently color them with few colors.  We also consider the problem
of efficiently coloring light tournaments, which is more technical and
is therefore deferred to Section \ref{sec:light2}.

\subsection{2-Colorable Tournaments}\label{sec:2color}

A tournament $T=(V,A)$ is {\em 2-colorable} if $\vec\chi(T) = 2$, and
a $2$-coloring of tournament $T$ is a partition of $V$ into two vertex
sets, $V_1$ and $V_2$, such that $T[V_1]$ and $T[V_2]$ are each
transitive.  In this section, our goal is to prove Theorem
\ref{thm:2-col}.

\begin{theorem}\label{thm:2-col}
  A 2-colorable tournament $T$ can be colored using ten colors in
  polynomial time.
  \end{theorem}

We say an arc $uv$ in $A$ is {\em heavy} if there exist three vertices
$a,b,c \in N(uv)$ which form a triangle $abc$.  If a tournament
contains no heavy arcs, then it is {\em light}.  We will use the
following observation.

\begin{observation}\label{obs:light}
There is a polynomial-time algorithm to partition a 2-colorable
tournament $T$ into two light subtournaments $T_1$ and $T_2$.
\end{observation}

This observation appears in
\cite{alon1996coloring,chen2005coloring,krivelevich2001approximating}
where it is stated more generally for 2-colorable 3-uniform
hypergraphs.  We include a proof here for completeness.

\begin{lemma}\label{lem:heavy}
In a 2-coloring of a tournament $T$, each heavy arc must be 2-colored.
\end{lemma}

\begin{proof}
If $u$ and $v$ are both, say, blue, then each vertex in $N(uv)$ would
be red, forcing a triangle in $N(uv)$ to be all red (i.e.,
monochromatic), which is not possible in a 2-coloring.
\end{proof}

\begin{corollary}\label{cor:heavy}
In a 2-colorable tournament, the heavy arcs form a bipartite graph. 
\end{corollary}

Now we can prove Observation \ref{obs:light}.

\begin{proof}[Proof of Observation \ref{obs:light}]
All heavy arcs can be easily detected.  By Corollary \ref{cor:heavy},
the set of heavy arcs forms a bipartite graph.  The vertex set of this
bipartite graph can be colored with two colors (red and blue), such
that the tournament induced by each color does not contain a heavy
arc.  Then we partition the vertices into two sets one containing all
the blue vertices and the other containing all the red vertices.  The
uncolored vertices can go in either set.  Since neither of these sets
contains any heavy arcs, this yields a partition the vertices of a
2-colorable tournament into two light subtournaments.
\end{proof}

Theorem \ref{thm:2-col} 
will follow from Observation \ref{obs:light} and the following theorem.

\begin{theorem}\label{thm:2-col-light}
A 2-colorable light tournament $T$ can be colored
with five colors in polynomial time.
\end{theorem}

Our goal is to use Lemma \ref{lem:eff_loc_to_glob} to prove Theorem
\ref{thm:2-col-light}.  In other words, we want to show that a 2-colorable
light tournament has a $(1,1)$-vertex chain that can be found efficiently.  We first prove a
useful claim.

\begin{lemma}\label{lem:min_max_vertex-k}
In a $k$-colorable tournament $T$, there exist distinct
vertices $u$ and $w$ such that $N^+(u) \cup N^-(w)$ is
$(k-1)$-colorable.
\end{lemma}

\begin{proof}
Since $T=(V,A)$ is $k$-colorable, there exist $k$ transitive sets
$X_1, \ldots, X_k$ such that $V= \cup_{i=1}^k X_i$. Then take $u$ to
be the vertex in $X_1$ that has only incoming arcs from other vertices
in $X_1$ (i.e., the sink vertex for $X_1$). Similarly, take $w$ to be
the vertex in $X_1$ that has only outgoing arcs to other vertices in
$X_1$ (i.e., the source vertex for $X_1$).  The out-neighborhood of
$u$ and the in-neighborhood of $w$ are both subsets of $V
\setminus{X_1}$, and thus so is their union, which is therefore
$(k-1)$-colorable.
\end{proof}

\begin{lemma}\label{lem:1vc}
A 2-colorable, light tournament $T$ contains a
$(1,1)$-vertex chain, which can be found in polynomial time.
\end{lemma}

\begin{proof}
By Lemma \ref{lem:min_max_vertex-k}, there exist $u$ and $w$ such that
$N^+(u) \cup N^-(w)$ is transitive.  To find them, we can test the
transitivity of $N^+(u) \cup N^-(w)$ for every pair of vertices in
$T$.  Then we simply need to find a shortest path from $u$ to $w$,
which can be done in polynomial time.  Let $k$ denote the length of the path.
Set $v_0 = u$ and $v_k = w$, and let $(v_i)_{1\leq i \leq k-1}$ be the remaining
vertices in the path.  Since
$\vec\chi(N(e)) \leq 1$ for every arc $e$ in a light tournament, we have $\vec\chi(N(e_i)) \leq 1$ for every arc in the arc chain corresponding to this vertex chain.  Moreover, notice that in this case, $T[N(e)]$ can be efficiently colored with one color.
\end{proof}

\begin{proof}[Proof of Theorem \ref{thm:2-col-light}]
The proof of Theorem \ref{thm:2-col-light} follows from Lemma
\ref{lem:1vc} and Lemma \ref{lem:eff_loc_to_glob}.
  \end{proof}

Now we are ready to prove the main theorem of this section.

\begin{proof}[Proof of Theorem \ref{thm:2-col}]
The proof of Theorem \ref{thm:2-col} follows from Observation \ref{obs:light} and Theorem \ref{thm:2-col-light}.\end{proof}

\paragraph{Certificates of Non-$2$-Colorability.}

In Section \ref{sec:2color}, we presented an algorithm to color a
2-colorable tournament with ten colors.  Suppose we run this algorithm
on an arbitrary tournament $T$ (e.g., one that is {\emph {not}}
2-colorable).  Then our algorithm will either color $T$ with ten colors or it
will produce at least one certificate that $T$ is not 2-colorable.  A
certificate will have the following form: either a) there is an odd
cycle of heavy arcs in $T$, or b) for every ordered pair of vertices
$(u,v)$, the subtournament $T[N^+(u) \cup N^-(v)]$ is not transitive.
In particular, an 11-chromatic tournament must contain such a
certificate.

\subsection{3-Colorable Tournaments}\label{sec:3color}

Coloring 3-colorable tournaments turns out to be closely related to
coloring 3-colorable graphs.  This seems surprising since the
techniques for 3-colorable graphs were applied to coloring 2-colorable
3-uniform hypergraphs, which are a generalization of 2-colorable
tournaments.

We will first show that we can adapt ideas of
\cite{wigderson1983improving} and \cite{blum1994new} to the problem of
coloring 3-colorable tournaments by using our algorithm for coloring
2-colorable tournaments with ten colors as a subroutine.

\begin{lemma}
A $3$-colorable tournament $T$ can be colored with $O(\sqrt{n})$ colors in
polynomial time.
\end{lemma}

\begin{proof}
A straightforward adaptation of Lemma 1 in \cite{blum1994new} implies
that if we can efficiently find a transitive subtournament of size
$\Omega(\sqrt{m})$ in a 3-colorable tournament on $m$ vertices, then
we can efficiently color a 3-colorable tournament on $n$ vertices with
$O(\sqrt{n})$ colors.  (We can substitute ``transitive set'' for
independent set in the definition of ``Type 1 Progress''
in \cite{blum1994new}.)  It remains to show that we can always find a
transitive set of size $\Omega(\sqrt{n})$ in a 3-colorable tournament
$T$ on $n$ vertices.

Let $T$ be a $3$-colorable tournament on $n \geq 3$ vertices.
Notice that $T$ has a vertex whose out-neighborhood is
2-colorable.  (In fact, it has at least three such vertices.)  To see this,
consider any $3$-coloring of $T$.  Consider a transitive
subtournament corresponding to one of the three colors.  Observe that
its sink vertex has outgoing arcs only towards the other two colors.

For any vertex, if its out-neighborhood is 2-colorable, we can color
its out-neighborhood with ten colors by Theorem \ref{thm:2-col}.  So we
can run the algorithm for the out-neighborhood of every vertex,
and the algorithm will successfully produce a 10-coloring of the
out-neighborhood of at least one vertex.
Therefore, if the minimum outdegree is at least $\sqrt{n}$, we find a
transitive set of size at least $\sqrt{n}/10$.

On the other hand, if
the minimum outdegree is smaller than $\sqrt{n}$, we will make
progress another way.  In this case, let $u$ be a vertex with
outdegree smaller than $\sqrt{n}$.  Then, we add $u$ to a set $S$, and
continue the algorithm on the subtournament of $T$ induced on
$V\setminus{N^+[u]}$.  We continue this until we find a transitive
subtournament of size at least $\sqrt{n}/20$ or until we have removed
half the vertices.  In the first case, we will have found a transitive
set of size $\Omega(\sqrt{n})$, and in the second case, the set $S$
will be transitive, and also of size $\Omega(\sqrt{n})$.
\end{proof}

We now show how to use the decomposition of Section \ref{sec:decomp} to get a
coloring with fewer colors based on a reduction to coloring
3-colorable graphs.

\begin{theorem}\label{thm:3-col_reduction}
If there is a polynomial-time algorithm to color a 3-colorable graph
$G$ with $k$ colors, then there is a polynomial-time algorithm to
color a 3-colorable tournament with $50k$ colors.
\end{theorem}

\begin{proof}
Let $T=(V,A)$ be a 3-colorable tournament. For every arc $e \in A$,
try coloring $N(e)$ with ten colors using Theorem \ref{thm:2-col}.  If
the algorithm fails, the neighborhood of the edge is not 2-colorable,
and thus the edge is not monochromatic in any 3-coloring.  Let $F
\subset E$ denote the set of arcs whose neighborhoods cannot be
colored with ten colors using our algorithm.  Ignore the direction of
the arcs in $F$ and consider the graph $G=(V,F)$.  This graph must be
3-colorable, since no arc in $F$ is monochromatic in any 3-coloring of
$T$.

Now let us show that from a coloring of $G$ with $k$ colors, we can
obtain a coloring of $T$ with $50k$ colors.  Consider a coloring of
the graph $G=(V,F)$ and let $V_i$ be the vertices colored with color
$i$ in this coloring.  Consider the induced subtournament $T' =
T[V_i]$; it has no arc in $F$ and thus the neighborhood of every arc
in this tournament can be colored efficiently with ten colors.
Furthermore, by Lemma \ref{lem:min_max_vertex-k} and Theorem
\ref{thm:2-col}, there are vertices $u$ and $v$ in $T'$ such that
$N^+_{T'}(u) \cup N_{T'}^-(v)$ is efficiently 10-colorable.  So by
Lemma \ref{lem:eff_loc_to_glob}, we can efficiently color $T'$ with 50
colors.  We can do this for the subtournament $T[V_i]$ for each of the
$i$ colors used to color $G$.
\end{proof}

Combining this lemma with the approximation algorithm from
\cite{kawarabayashi2017coloring}, which colors a 3-colorable graph
with fewer than $\widetilde{O}(n^{\frac{1}{5}})$ colors,
we obtain the same
asymptotic bound for 3-colorable tournaments.

\begin{corollary}
A 3-colorable tournament $T$ can be colored with 
$\widetilde{O}(n^{0.19996})$ colors in polynomial time.
\end{corollary}

We can extend Theorem \ref{thm:3-col_reduction} to a more general
case.

\begin{lemma}
Let $g$ be a function such that we can efficiently color
a $k$-colorable graph with
$g(k)$ colors, and let $f$ be a function such that we can efficiently color
a $k$-colorable tournament with $f(k)$ colors.  
Then $f(k) \leq 5\cdot f(k-1)
\cdot g(k)$.
\end{lemma}

\begin{proof}
We use the same reduction as in the proof of Theorem
\ref{thm:3-col_reduction}, but now $F$ is the set of arcs whose
neighborhoods cannot be efficiently $f(k-1)$-colored.  Then each $V_i$
in $G$ is colored with $5 \cdot f(k-1)$ colors.  So we need a total of
$5 \cdot f(k-1) \cdot g(k)$ colors.
\end{proof}

\section{Hardness of Approximate Coloring in Tournaments}\label{sec:hardness}

In this section, we examine the hardness of approximate coloring of
tournaments.  \cite{chen2007min} showed that deciding if a tournament
can be 2-colored is \NP-hard.  For completeness, we provide a
simplified (though similar) proof of this result in Appendix
\ref{sec:base_hardness}.  Later, \cite{fox2019removal} proved that for
any $k$, it is \NP-hard to decide if a tournament is $k$-colorable.

We will first improve upon these \NP-hardness results and then show
hardness of coloring $k$-colorable tournaments for $k\geq 3$ with
$O(1)$ colors under the $d$-To-1 conjecture.  The $d$-To-1 conjecture
was first introduced by Khot alongside the famous Unique Games
conjecture \cite{khot2002power}, and has since been used to show
hardness of coloring $3$-colorable graphs with $O(1)$ colors
\cite{guruswami2020d}.

\subsection{\NP-Hardness of Approximate Coloring of $k$-Colorable
  Tournaments}\label{sec:2-color-hardness}

It was shown previously that it is \NP-hard to color a 2-colorable
tournament with two colors~\cite{chen2007min,fox2019removal}.  We prove
the following stronger theorem.

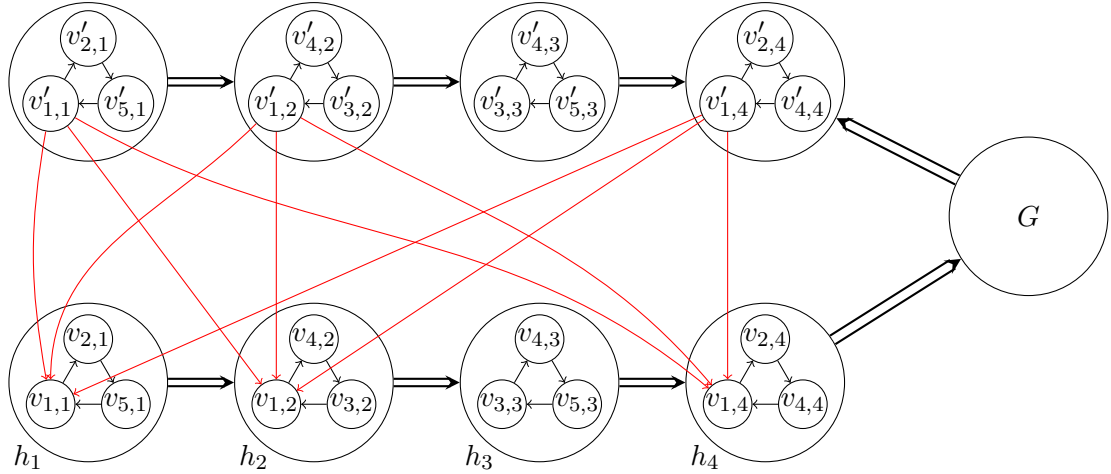
\begin{figure}
\centering
\begin{tikzpicture}
	  \tikzstyle{smallvertex}=[circle,draw,minimum size=10pt,inner sep=0pt]
	  \tikzstyle{bigvertex}=[circle,draw,minimum size=40pt,inner sep=0pt]
	  \tikzstyle{group}==[circle,draw,minimum size=60pt,inner sep=0pt]

	  \node[smallvertex] (v1) at (0,0) {$v_{1,1}$};
	  \node[smallvertex] (v2) at (0.5,0.866) {$v_{2,1}$};
	  \node[smallvertex] (v3) at (1,0) {$v_{5,1}$};
	  \node[group] (g1) at (0.5,0.289) {$$};
	  \node at (-0.3,-0.7) {$e_1$};
	  \draw [->] (v1) to (v2);
	  \draw [->] (v2) to (v3);
	  \draw [->] (v3) to (v1);

	  \begin{scope}[xshift=3cm]
	  \node[smallvertex] (v4) at (0,0) {$v_{1,2}$};
          \node[smallvertex] (v5) at (0.5,0.866) {$v_{4,2}$};
          \node[smallvertex] (v6) at (1,0) {$v_{3,2}$};
          \node[group] (g2) at (0.5,0.289) {$$};
	  \node at (-0.3,-0.7) {$e_2$};
          \draw [->] (v4) to (v5);
          \draw [->] (v5) to (v6);
          \draw [->] (v6) to (v4);
	  \end{scope}

          \begin{scope}[xshift=6cm]
          \node[smallvertex] (v7) at (0,0) {$v_{3,3}$};
          \node[smallvertex] (v8) at (0.5,0.866) {$v_{4,3}$};
          \node[smallvertex] (v9) at (1,0) {$v_{5,3}$};
          \node[group] (g3) at (0.5,0.289) {$$};
	  \node at (-0.3,-0.7) {$e_3$};
          \draw [->] (v7) to (v8);
          \draw [->] (v8) to (v9);
          \draw [->] (v9) to (v7);
          \end{scope}

          \begin{scope}[xshift=9cm]
          \node[smallvertex] (v10) at (0,0) {$v_{1,4}$};
          \node[smallvertex] (v11) at (0.5,0.866) {$v_{2,4}$};
          \node[smallvertex] (v12) at (1,0) {$v_{4,4}$};
          \node[group] (g4) at (0.5,0.289) {$$};
	  \node at (-0.3,-0.7) {$e_4$};
          \draw [->] (v10) to (v11);
          \draw [->] (v11) to (v12);
          \draw [->] (v12) to (v10);
          \end{scope}

	  \draw [->,>=stealth,thick,double distance=2pt] (g1) to (g2);
	  \draw [->,>=stealth,thick,double distance=2pt] (g2) to (g3);
	  \draw [->,>=stealth,thick,double distance=2pt] (g3) to (g4);

	  \begin{scope}[yshift=4cm]
          \node[smallvertex] (u1) at (0,0) {$v'_{1,1}$};
          \node[smallvertex] (u2) at (0.5,0.866) {$v'_{2,1}$};
          \node[smallvertex] (u3) at (1,0) {$v'_{5,1}$};
          \node[group] (h1) at (0.5,0.289) {$$};
          \draw [->] (u1) to (u2);
          \draw [->] (u2) to (u3);
          \draw [->] (u3) to (u1);

          \begin{scope}[xshift=3cm]
          \node[smallvertex] (u4) at (0,0) {$v'_{1,2}$};
          \node[smallvertex] (u5) at (0.5,0.866) {$v'_{4,2}$};
          \node[smallvertex] (u6) at (1,0) {$v'_{3,2}$};
          \node[group] (h2) at (0.5,0.289) {$$};
          \draw [->] (u4) to (u5);
          \draw [->] (u5) to (u6);
          \draw [->] (u6) to (u4);
          \end{scope}

          \begin{scope}[xshift=6cm]
          \node[smallvertex] (u7) at (0,0) {$v'_{3,3}$};
          \node[smallvertex] (u8) at (0.5,0.866) {$v'_{4,3}$};
          \node[smallvertex] (u9) at (1,0) {$v'_{5,3}$};
          \node[group] (h3) at (0.5,0.289) {$$};
          \draw [->] (u7) to (u8);
          \draw [->] (u8) to (u9);
          \draw [->] (u9) to (u7);
          \end{scope}

          \begin{scope}[xshift=9cm]
          \node[smallvertex] (u10) at (0,0) {$v'_{1,4}$};
          \node[smallvertex] (u11) at (0.5,0.866) {$v'_{2,4}$};
          \node[smallvertex] (u12) at (1,0) {$v'_{4,4}$};
          \node[group] (h4) at (0.5,0.289) {$$};
          \draw [->] (u10) to (u11);
          \draw [->] (u11) to (u12);
          \draw [->] (u12) to (u10);
          \end{scope}

          \draw [->,>=stealth,thick,double distance=2pt] (h1) to (h2);
          \draw [->,>=stealth,thick,double distance=2pt] (h2) to (h3);
          \draw [->,>=stealth,thick,double distance=2pt] (h3) to (h4);
	  \end{scope}

	  \node[group] (G) at (13,2.5) {$G$};

          \draw [->,>=stealth,thick,double distance=3pt] (g4) to (G);
          \draw [->,>=stealth,thick,double distance=3pt] (G) to (h4);
          \draw [->,red, out=-100, in=100] (u1) to (v1);
          \draw [->,red, out=-135, in=90] (u4) to (v1);
          \draw [->,red] (u10) to (v1);
          \draw [->,red] (u1) to (v4);
          \draw [->,red] (u4) to (v4);
          \draw [->,red] (u10) to (v4);
          \draw [->,red, out = -30, in = 140] (u1) to (v10);
          \draw [->,red, out = -30, in = 130] (u4) to (v10);
          \draw [->,red] (u10) to (v10);

\end{tikzpicture}
\caption{Construction of $T$ from a 3-uniform hypergraph $\Hd$.  (The
  full description of the construction is in the Proof of Theorem
  \ref{thm:2-hardness}.)  The downwards edges in red are drawn
  only for vertex $v_1$ in $\Hd$, but there is an arc from any vertex
  $v'_{a,i}$ towards all vertices $v_{a,j}$ for any $j$. The remaining
  arcs all go upwards from the vertices $v_{a,i}$ towards the vertices
  $v'_{b,j}$ for $a \neq b$.}\label{fig:2-hard}
\end{figure}

\begin{theorem}\label{thm:2-hardness}
  It is \NP-hard to color a 2-colorable tournament with three colors.
\end{theorem}

\begin{proof}
\cite{dinur2005hardness} proved that it is \NP-hard to color a
3-uniform 2-colorable hypergraph with $c$ colors for any constant $c$.
In particular, it is \NP-hard to color a 3-uniform 2-colorable
hypergraph with six colors.  We will show that if we can efficiently
color a 2-colorable tournament with three colors, then we can efficiently
color a 2-colorable 3-uniform hypergraph with six colors, which we have
just noted is an \NP-hard problem.

Let $\Hd = (\V, \Ed)$ be a 3-uniform hypergraph.  In
\cite{fox2019removal} and \cite{chen2007min}, it is shown how to
efficiently construct a tournament $G$ such that $G$ is 2-colorable
iff $\Hd$ is $2$-colorable.  Moreover, given a 2-coloring of $G$, one
can efficiently construct a 2-coloring of $\Hd$.  Details can be found
in Appendix \ref{sec:base_hardness}.  Using this result, we will show
how to efficiently construct a new tournament $T=(V,A)$ such
that if $\Hd$ is 2-colorable, $T$ is also $2$-colorable.  Then we will
show that if we can efficiently color $T$ with three colors, we can
efficiently color $\Hd$ with six colors, which establishes the theorem.

We will start by defining a subtournament $H=(V_1,A_1)$ of $T$.  We
fix an arbitrary enumeration of the hyperedges $\Ed$.  For each edge,
$e_i \in \Ed$, where $e_i=(v_a,v_b,v_c)$, we add three vertices
$v_{a,i}$, $v_{b,i}$ and $v_{c,i}$ to $V_1$, and add to $A_1$ the arcs
$(v_{a,i},v_{b,i})$, $(v_{b,i},v_{c,i})$ and $(v_{c,i},v_{a,i})$ such
that these three vertices form a directed triangle.  We then add the
arcs from all the vertices $v_{a,i}$ towards all the vertices
$v_{b,j}$ for any $a,b,i,j$ with $i<j$.  We make a copy of $H$, that
we call $H'=(V_2,A_2)$, and add both to $T$.  We then add the
tournament $G$, and orient all arcs from vertices in $V_1$ towards
vertices of $G$, and all arcs from vertices of $G$ towards vertices in
$V_2$.  The only arcs we still need to orient are those between $V_1$
and $V_2$.  For this, we look at the vertices in $\V$ from which the
vertices of $T$ are derived; for $v_{a,i} \in V_1$ and $v'_{b,j} \in
V_2$, we add an arc from $v'_{b,j}$ to $v_{a,i}$ iff $a=b$ (i.e., if
they are derived from the same vertex of $\Hd$), and we add an arc
from $v_{a,i}$ to $v'_{b,j}$ otherwise. This completes the definition
of $T$, which is shown in Figure \ref{fig:2-hard}.

We will now establish that if $\Hd$ is $2$-colorable, so is $T$. Given
a $2$-coloring of $\Hd$, give all the vertices of $V_1$ the same color
as the vertex of $\Hd$ from which they are derived, and those in $V_2$
the opposite color of the vertex of $\Hd$ from which they are derived.
Finally color $G$ with the same two colors. Then any arc that goes
from $V_2$ to $V_1$ will be bicolored, and since all arcs are
oriented from $V_1$ towards $G$ and from $G$ towards $V_2$, there can
only be monochromatic triangles inside $V_1$, $V_2$ or $G$. However,
$G$ is properly $2$-colored and thus does not have any monochromatic
triangles. Furthermore, every triangle in $V_1$ and $V_2$ represents a
hyperedge of $\Hd$ and must therefore contain two vertices of
different colors.

It remains to show that if we can efficiently find a 3-coloring of
$T$, then we can construct a 6-coloring of $\Hd$.  Consider such a
3-coloring $C$ of $T$.  Notice that if $G$ uses two colors in $C$,
then we can recover a 2-coloring of $\Hd$.  So we can assume that $G$
uses three colors in $C$.  For every vertex $v_a \in \V$, consider the
set of vertices $S_a=\{v_{a,i}~ |~ \forall e_i \in \Ed \}$ and
$Q_a=\{v'_{a,i} ~|~ \forall e_i \in \Ed \}$.  A key property of our
construction is that in any $3$-coloring of $T$ in which $G$ uses
three colors, for each $v_a \in \V$, the set $S_a$ or the set $Q_a$
must be monochromatic. To see this, notice that if any vertex of $S_a$
has the same color as any vertex of $Q_a$, then they will form a
monochromatic triangle with a third vertex from $G$ that has the same
color (since $G$ is colored with at least three colors).  So if $S_a$
and $Q_a$ each use at least two out of three colors, then at least one
color appears in both $S_a$ and $Q_a$ resulting in a monochromatic
triangle.

Next we define a coloring $\cch$ of $\Hd$ as follows.  If $S_a$ is
monochromatic, then set $\cch(v_a)=C(S_a)$.  Otherwise, set
$\cch(v_a)=C(Q_a)+3$.  Now consider any hyperedge $(v_a,v_b,v_c)$ in
$\Ed$. If the three sets $S_a$, $S_b$ and $S_c$ are monochromatic,
then since there is a directed triangle $(v_{a,j},v_{b,j},v_{c,j})$ in
$H$ for some $j$, the three vertices cannot have the same color in
$C$, so they also do not all have the same color in $\cch$.  If none
of the three sets $S_a$, $S_b$ and $S_c$ are monochromatic, then the
sets $Q_a, Q_b$ and $Q_c$ are each monochromatic, so the same argument
applies.  Finally, without loss of generality we can suppose $S_a$ is
monochromatic but not $S_b$.  Then $v_a$ and $v_b$ do not have the
same color in $\cch$ by definition.  Therefore, no hyperedge of $\Hd$
can be monochromatic, and thus $\cch$ is a $6$-coloring of $\Hd$.
\end{proof}

Our goal is now to extend this hardness result to $k$-colorable
tournaments.  As noted in the previous proof, the main theorem of
\cite{dinur2005hardness} says that for any integer $c \geq 2$, it is
\NP-hard to color a 2-colorable 3-uniform hypergraph $\Hd$ with $c$
colors.  In fact, what they prove is that it is \NP-hard to decide
between the two cases: $\chi(\Hd) =2$ and $\chi(\Hd) > c$.  We use
this latter decision version to prove the following theorem.

\begin{theorem}\label{thm:hardness}
For any fixed positive integer $k \geq 2$, it is \NP-hard to color a $k$-colorable tournament with $2k-1$ colors.
\end{theorem}

Our proof uses an iterative construction given in the following lemma.

\begin{lemma}\label{lem:gadget}
Let $a,b,c,d$ be positive integers such that $a+b < c+d$.  Let $R_1$
and $R_2$ be two tournaments such that either (i) $\vec\chi(R_1) = a$
and $\vec\chi(R_2) = b$, or (ii) $\vec\chi(R_1) \geq c$ and
$\vec\chi(R_2) \geq d$.  Then we can efficiently construct a
tournament $R'$ with $\vec\chi(R') = a+b$ in case (i), or
$\vec\chi(R') \geq c+d$ in case (ii). 
\end{lemma}

\begin{proof}
The proof of the lemma uses the next two claims.

\begin{claim}\label{clm:rec_gadget1}
Let $a,b$ be positive integers.  Let $R_1$, $R_2$ and $R_3$ be three
tournaments such that $\vec\chi(R_1) = a$, $\vec\chi(R_2) = b$ and
$\vec\chi(R_3) = a+b$.  Then the tournament $R' = \Delta(R_1,R_2,R_3)$
has $\vec\chi(R') = a+b$.
\end{claim}

\begin{cproof}
By assumption, we can color $R_1$ with $a$ colors, $R_2$ with
$b$ (different) colors and $R_3$ with the same set of $a+b$ colors.
This dicoloring of $R'$ is proper since there is no monochromatic
triangle inside $R_1$, $R_2$ or $R_3$, and any triangle containing
vertices from $R_1$ and $R_2$ will have at least two different colors.
\end{cproof}

\begin{claim}\label{clm:rec_gadget2}
  Let $c,d,e$ be positive integers such that $e < c+d$.  Let $R_1$,
  $R_2$ and $R_3$ be three tournaments such that $\vec\chi(R_1) \geq
  c$, $\vec\chi(R_2) \geq d$ and $\vec\chi(R_3) \geq e$.  Then the
  tournament $R' = \Delta(R_1,R_2,R_3)$ has $\vec\chi(R') \geq e+1$.
\end{claim}

\begin{cproof}
Suppose $R'$ has a coloring with $e$ colors.  Since $c+d>e$, $R_1$ and
$R_2$ must share at least one color.  Furthermore, all $e$ colors are
used in $R_3$ by assumption.  So there must be a monochromatic
triangle since every triplet $(u,v,w)$ with $u \in R_1$, $v \in R_2$,
$w \in R_3$ forms a directed triangle. Thus, $\vec\chi(R') \geq e+1$.
\end{cproof}

We are now ready to prove the lemma, which we prove
by induction on $k$, where $a+b \leq k \leq c+d$.  For
the base case, when $k=a+b$, let $R_3$ be any tournament with
$\vec\chi(R_3) = a+b$.\footnote{For any $i$, we can construct a
tournament $S_i$ with $\vec\chi{(S_i)} = i$ on $2^i-1$ vertices.  To see
this, let $S_1$ be a single vertex and let $S_{i+1} = \Delta(S_1,S_i,S_i)$.}  By Claims \ref{clm:rec_gadget1} and
\ref{clm:rec_gadget2}, the tournament $R' = \Delta(R_1,R_2,R_3)$ has
$\vec\chi(R') = a+b$ in case (i) or $\vec\chi(R') \geq a+b+1$ in case (ii).

Now set $R_k' = R'$ and let $R_{k+1}' = \Delta(R_1,R_2,R_k')$.  By the
inductive hypothesis, we assume that $\vec\chi(R_k') = a+b$ in case
  (i) or $\vec\chi(R_k') \geq k+1$ in case (ii).  Thus, we will have
  $\vec\chi(R_{k+1}') = a+b$ in case (i) (by Claim \ref{clm:rec_gadget1}) or $\vec\chi(R_{k+1}') \geq
  k+2$ in case (ii) (by Claim \ref{clm:rec_gadget2}).  Finally, we set $R' = R_{c+d}'$.

Every iteration of the construction can be done in
time polynomial in the size of $R_1$ and $R_2$ for fixed values of $a,b,c,d$.
There are at most $c+d$ iterations.  Thus, $R'$ can be
constructed in time polynomial in the size of $R_1$ and $R_2$ and has size $|V(R')| \leq (c+d) \cdot
(|V(R_1)| + |V(R_2)|) + |V(R_{a+b})|$.
\end{proof}

We are now ready to prove Theorem \ref{thm:hardness}.

\begin{proof}[Proof of Theorem \ref{thm:hardness}]
We will show that if we can color a $k$-colorable tournament with
$2k-1$ colors, then we can decide whether a given 3-uniform hypergraph
$\Hd$ has $\chi(\Hd) = 2$ or $\chi(\Hd) \geq 7$, which we have noted
is an \NP-hard problem.

Given a 3-uniform hypergraph $\Hd$, we will prove by strong induction
that for every fixed $k$, we can efficiently construct a tournament
$T_k$ whose size is polynomial in $|V(\Hd)|$, such that if
$\chi(\Hd)=2$ then $\vec\chi(T_k) = k$, and if $\chi(\Hd) \geq 7$ then
$\vec\chi(T_k) \geq 2k$.

For $k=2$, we refer to the tournament, call it $T_2$, constructed in
the proof of Theorem \ref{thm:2-hardness}.  We showed that if
$\chi(\Hd) = 2$, then $\vec\chi(T_2) = 2$ and if $\chi(\Hd) \geq 7$,
then $\vec\chi(T_2) \geq 4$.  For $k=3$, let $T_3 = \Delta(T_2, T_2,
T_2)$.  If $\chi(\Hd)=2$, coloring the first copy with colors $1$,
$2$, the second with colors $2$ and $3$, and the third with colors $3$
and $1$ yields a 3-coloring.  The tournament $T_3$ is not 2-colorable
since in any 2-coloring of $T_3$, each copy of $T_2$ must use the same
two colors, which would result in a monochromatic directed triangle.
If $\chi(\Hd)\geq 7$, as noted, $T_2$ has chromatic number at least
$4$. Therefore, in any $5$-coloring of $T_3$, there are two colors
that must be used in each copy of $T_2$, which would lead to a
monochromatic directed triangle. Therefore, $\vec\chi(T_3) \geq 6$.

Now our induction hypothesis is: For every $h$ such that $3 \leq h
\leq k$, there exists a tournament $T_h$ of size polynomial in $h$ and
$|V(\Hd)|$ such that if $\chi(\Hd)=2$, $\vec\chi(T_h) = h$, and if
$\chi(\Hd) \geq 7$, $\chi(T_h) \geq 2h$.  We will show that there
exists a tournament $T_{k+1}$ of size polynomial in $|V(\Hd)|$ such
that if $\chi(\Hd)=2$, $\vec\chi(T) = k+1$, and if $\chi(\Hd) \geq 7$,
$\vec\chi(T) \geq 2 (k+1)$.

Consider the two tournaments $T_{\lfloor \frac{k+1}{2} \rfloor}$,
$T_{\lceil \frac{k+1}{2} \rceil}$, which exist by the induction
hypothesis.  These obey the conditions of Lemma \ref{lem:gadget},
where $a = {\lfloor \frac{k+1}{2} \rfloor}$, $b = {\lceil
  \frac{k+1}{2} \rceil}$, $c = 2a$ and $d=2b$.  Thus, by Lemma
\ref{lem:gadget}, there exists a tournament, $T_{k+1}$, such that if
$\chi(\Hd)=2$, $\vec\chi(T_{k+1}) = k+1$, and if $\chi(\Hd) \geq 7$,
$\vec\chi(T_{k+1}) \geq 2(\lceil \frac{k+1}{2} \rceil + \lfloor
\frac{k+1}{2} \rfloor) = 2(k+1)$.  This concludes the induction.
\end{proof}

\subsection{Reduction from Coloring Graphs to Coloring Tournaments}

In Section \ref{sec:3color}, we showed that if we can color a
3-colorable graph with $k$ colors, then we can color a 3-colorable
tournament with $50k$ colors.  In this section, we give a reduction in
the other direction. Specifically, for $\ell > k$, we show that the
problem of deciding if a graph is $k$-colorable or has chromatic
number at least $\ell$ has a polynomial-time reduction to the problem
of deciding if a tournament is $k$-colorable or has dichromatic number
at least $\ell$.  A corollary of this reduction is hardness of
coloring tournaments under the $d$-To-1 Conjecture of
Khot~\cite{khot2002power}; \cite{guruswami2020d} showed that assuming
the $d$-To-1 Conjecture, it is hard to color 3-colorable graphs with
$O(1)$ colors, and using our reduction, we can extend this hardness to
tournaments.

\begin{theorem}\label{thm:graph_to_tourn}
Suppose that for any constants $\ell > k \geq 3$, and for any
tournament $T$, we can efficiently decide if $\vec\chi(T) = k$ or
$\vec\chi(T) > \ell$.  Then for any graph $G$, we can
efficiently decide if $\chi(G) = k$ or $\chi(G) > \ell$.
\end{theorem}

We start by proving the following lemma that presents the building
block of the reduction.

\begin{lemma}\label{lem:induction}
  Let $c \geq 3$ be an integer, 
let $G=(V_G,E_G)$ be a graph and let $T=(V_T,A_T)$ a tournament such that
$\vec\chi(T)=k$ when $\chi(G) = k$, and $\vec\chi(T) \geq
min(\chi(G),c)$ when $\chi(G) > k$.  We can build a new tournament
$U=(V_U,A_U)$ such that $\vec\chi(U) = k$ when $\chi(G) = k$, and
$\vec\chi(U)\geq min(\chi(G),c+1)$ when $\chi(G) > k$.
\end{lemma}

\begin{figure}
	\centering
\begin{tikzpicture}[->,>=stealth',auto,node distance=2cm,thick,main node/.style={circle,draw,font=\sffamily\Large\bfseries}]
  \tikzstyle{smallvertex}=[circle,draw,minimum size=20pt,inner sep=0pt]
  \tikzstyle{bigvertex}=[circle,draw,minimum size=40pt,inner sep=0pt]
  
  \node[smallvertex] (v1) at (-2,5) {$v_1$};
  \node[smallvertex] (v2) [right of=v1] {$v_2$};
  \node[smallvertex] (v3) [right of=v2] {$v_3$};
  \node[smallvertex] (v4) [right of=v3] {$v_4$};
  \node[smallvertex] (v5) [right of=v4] {$v_5$};

  \node[bigvertex] (w1) at (-1,0) {$T_1$};
  \node[bigvertex] (w2) at (1,0) {$T_2$};
  \node[bigvertex] (w3) at (3,0) {$T_3$};
  \node[bigvertex] (w4) at (5,0) {$T_4$};

  \draw [->,>=stealth,thick,double distance=2pt] (v1) to (w1);
  \draw [->,>=stealth,thick,double distance=2pt] (v2) to (w2);
  \draw [->,>=stealth,thick,double distance=2pt] (v3) to (w3);
  \draw [->,>=stealth,thick,double distance=2pt] (v4) to (w4);
  \draw [->,>=stealth,thick,double distance=2pt] (v1) to (w2);
  \draw [->,>=stealth,thick,double distance=2pt] (v2) to (w3);
  \draw [->,>=stealth,thick,double distance=2pt] (v3) to (w4);
  \draw [->,>=stealth,thick,double distance=2pt] (v1) to (w3);
  \draw [->,>=stealth,thick,double distance=2pt] (v2) to (w4);
  \draw [->,>=stealth,thick,double distance=2pt] (v1) to (w4);
  \draw [->,>=stealth,thick,double distance=2pt] (w1) to (v2);
  \draw [->,>=stealth,thick,double distance=2pt] (w2) to (v3);
  \draw [->,>=stealth,thick,double distance=2pt] (w3) to (v4);
  \draw [->,>=stealth,thick,double distance=2pt] (w4) to (v5);
  \draw [->,>=stealth,thick,double distance=2pt] (w1) to (v3);
  \draw [->,>=stealth,thick,double distance=2pt] (w2) to (v4);
  \draw [->,>=stealth,thick,double distance=2pt] (w3) to (v5);
  \draw [->,>=stealth,thick,double distance=2pt] (w1) to (v4);
  \draw [->,>=stealth,thick,double distance=2pt] (w2) to (v5);
  \draw [->,>=stealth,thick,double distance=2pt] (w1) to (v5);
  \draw [->,>=stealth,thick,double distance=2pt] (w1) to (w2);
  \draw [->,>=stealth,thick,double distance=2pt] (w2) to (w3);
  \draw [->,>=stealth,thick,double distance=2pt] (w3) to (w4);
    \path[every node/.style={font=\sffamily\small},red,dashed]
    (v3) edge node [right] {} (v2)
    (v4) edge node [right] {} (v3)
    (v3) edge node [right] {} (v2)
    (v4) edge[bend right, in = -135] node [right] {} (v1)
    (v4) edge[bend right] node [right] {} (v2)
    (v5) edge[bend right, out = -45, in = -135] node [right] {} (v1);
    \path[every node/.style={font=\sffamily\small},blue]
    (v1) edge node [right] {} (v2)
    (v4) edge node [right] {} (v5)
    (v1) edge[bend left] node [right] {} (v3)
    (v2) edge[bend left, out = 45] node [right] {} (v5)
    (v3) edge[bend left] node [right] {} (v5);
\end{tikzpicture}
\caption{Construction of the tournament $U$ from a graph $G$ on five vertices. The dashed red edges are those present in $G$ and all go backwards, whereas the remaining edges are blue and go forwards. \label{fig:c-hardness}}
\end{figure}
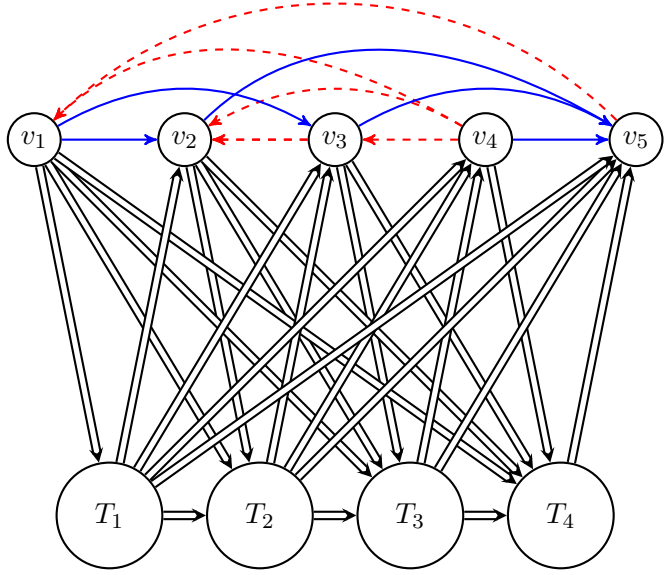

\begin{proof}
Let $n_G = |V_G|$ and let $(T_i)_{1\leq i \leq n_G-1}$ be copies of
$T$.  Let $T_i=(V_i,A_i)$. Then $V_U := (\cup_{1\leq i \leq n_G-1}
V_i) \cup V_G$.  Fix an arbitrary ordering of the vertices in
$V_G$. To build $A_U$, add the arc from $v_j$ to $v_i$ where $i < j$
if $(v_i,v_j) \in E_G$, and the arc from $v_i$ to $v_j$ if
$(v_i,v_j) \notin E_G$. The resulting tournament induced on
the vertices of $V_G$ is said to have $G$ as a {\emph {backedge
  graph}}.  Next we add all the arcs from $v_i$ to all vertices of
$T_j$ for every $i\leq j$, and the arcs from every vertex of $T_i$ to
$v_j$ for all $i<j$. Finally, we add the arcs from any vertex of $T_i$
to any vertex of $T_j$ for every $i<j$. This concludes the
construction of $U$, which is depicted in Figure \ref{fig:c-hardness}.

Suppose $\chi(G) = k$. Then let us show that $\vec\chi(U) = k$. In
this case, $\vec\chi{(T)} = k$ by assumption.  We take a $k$-coloring
of $G$ and a $k$-coloring of $T$ and color the vertices in $U$ (i.e.,
use the $k$-coloring of $G$ for $V_G$ and the $k$-coloring of $T$ for
$V_{i}$ for all $1 \leq i \leq n_G-1$).  Notice that all arcs that are
backwards with respect to the order $v_1 \rightarrow T_1 \rightarrow
v_{2} \rightarrow ... \rightarrow v_i \rightarrow T_i \rightarrow ...
\rightarrow T_{|V_G|-1} \rightarrow v_{|V_G|}$ are bicolored.  To see
this, observe that arcs from $v_j$ to $v_i$ for $j > i$ belong to
$E_G$ and are therefore bicolored, and by construction, there are no
arcs from $v_j$ to $T_i$ nor from $T_j$ to $T_i$ for $j>i$.  Thus,
there can only possibly be monochromatic triangles within $T_i$, but
these sets are properly colored.  Therefore, this is a proper
dicoloring of the tournament $U$ and $\vec\chi(U)=k$.

Let us now prove that when $\chi(G) > k$, we have $\vec\chi(U)\geq
\min\{\chi(G),c+1\}$.  By assumption, we have $\vec\chi(T) \geq
\min\{\chi(G),c\}$ in this case.  Thus, if $c \geq \chi(G)$, then the
claim is true, since $T$ is a subtournament of $U$.  So let us
consider the case in which $c < \chi(G)$.  Then given a coloring of
$U$ with $c$ colors, there must be a monochromatic edge $(v_i,v_j)$ in
$G$.  Assuming without loss of generality that $i<j$, there is a
monochromatic arc from $v_j$ to $v_i$ in $U$.  Furthermore, since
$\vec\chi(T) \geq c$, there must be some vertex of $T_i$ that has the
same color as $v_i$ and $v_j$.  Since all vertices in $T_i$ form a
directed triangle with $v_i$ and $v_j$, this means that there is a
monochromatic triangle in $U$, which is a contradiction.
\end{proof}

We can now prove Theorem \ref{thm:graph_to_tourn} by a simple
induction.

\begin{proof}[Proof of Theorem \ref{thm:graph_to_tourn}]
Let $G=(V_G,E_G)$ be a graph.  For all $c \geq k$, we will build a
tournament $T_c=(V_{T_c},A_{T_c})$ by induction such that if $\chi(G)
= k$, then $\vec\chi(T_c) = k$, and if $\chi(G) \geq \ell$, then
$\vec\chi(T_c) \geq \min\{\chi(G),c\}$.

For $c=k$, any $k$-colorable tournament, say $S_k$, satisfies the
conditions.  For $k+1$, we obtain $T_{k+1}$ by applying Lemma
\ref{lem:induction} where $T_c$ is $S_k$ and $T_{c+1}$ is $U$.
Suppose by induction that there is a tournament $T_c$ satisfying the
conditions for a constant $c$.  Let us show that there is a tournament
$T_{c+1}$ that satisfies these same conditions for $c+1$.  This
follows from Lemma \ref{lem:induction} where $T_c$ is $T$, and
$T_{c+1}$ is $U$.

The tournament $T_\ell$ has size $|V_{T_{\ell}}| = O(|2^k \cdot
V_G|^{\ell})$, which is polynomial for fixed $\ell$.  Clearly, the
time to build $T_{\ell}$ is polynomial in its size.  Furthermore, if
$\chi(G) = k$ then $\vec\chi(T_\ell) = k$, and if $\chi(G) \geq \ell$,
then $\vec\chi(T_\ell) \geq \min\{\chi(G),\ell\}$.  Thus, if we can
efficiently decide if $T_\ell$ has chromatic number $k$ or at least
$\ell$, then we can also efficiently decide if $G$ has chromatic
number $k$ or at least $\ell$.
\end{proof}

Under the $d$-to-$1$-conjecture~\cite{guruswami2020d} (and under
another conjecture discussed in \cite{dinur2009conditional}), for any
constant $c \geq 4$, it is \NP-hard to decide if a graph is
3-colorable or if it has chromatic number at least $c$.  This implies
equivalent hardness for coloring $3$-colorable tournaments and for
coloring $k$-colorable tournaments for $k \geq 3$ (since any
3-colorable tournament is also $k$-colorable when $k\geq 3$).

\begin{corollary}\label{thm:dto1-hardness}
Let $\ell > k \geq 3$ be any constants.  Then if the $d$-To-$1$
conjecture is true, we cannot efficiently decide if a tournament $T$
has $\vec\chi(T) = k$ or $\vec\chi(T) > \ell$.
\end{corollary}

Notice that if stronger hardness (for example constant hardness under
the $\P \neq \NP$ assumption) were established for approximate
coloring of $3$-colorable graphs, then this reduction would provide
stronger hardness results for $3$-colorable tournaments.  This would
hold up to constant hardness, after which the blowup of the size of
the tournament in the construction would be more then polynomial.

Finally, we consider the hardness of the problem of coloring general
tournaments.  Coloring digon-free digraphs has been shown to be
\NP-hard to approximate within a factor of $n^{1/2-\eps}$
\cite{feder2019complexity}. This proof can easily be extended to the
case of tournaments, which provides the following theorem.

\begin{restatable}{theorem}{genhardness}\label{thm:gen_hardness}
It is \NP-hard (under randomized reductions) to approximate the
dichromatic number of tournaments within a factor of $n^{1/2-\delta}$
for any $0 < \delta < 1/2$.
\end{restatable}

The proof of this Theorem is given in Appendix \ref{sec:hardness_gen_app}.

\section{Light Tournaments}\label{sec:light2}

Light tournaments are exactly those which do not contain the hero
$\Delta(C_3,1, 1)$, where $C_3$ is a directed triangle and `$1$' is a
single vertex.  \cite{berger2013tournaments} proved that light
tournaments have constant chromatic number, but they did not state a
precise constant, and their proof is not algorithmic.  A careful
modification of their approach can be used to give an algorithmic
proof that this constant is around 35.  In this section, our goal is
to prove the following theorem.
\begin{theorem}\label{thm:bound_light}
  Let $T$ be a light tournament.
  Then we can color $T$ with at most eight colors in polynomial time.
\end{theorem}

\begin{lemma}\label{lem:endpoints}
  Let $T$ be a light tournament.  Then we can find $u,v$ such that:
  \begin{itemize}
  \item[(i)] $T[N^+(u)]$ can be colored with three colors in polynomial time,
  \item[(ii)] $T[N^-(v)]$ can be colored with three colors in polynomial time, and 
    \item[(iii)] $T[N^-(v) \cup N^+(u)]$ can be colored with five colors in polynomial time.
    \end{itemize}
  \end{lemma}

Assuming Lemma \ref{lem:endpoints}, we can prove Theorem
\ref{thm:bound_light}.  

\begin{proof}[Proof of Theorem \ref{thm:bound_light}]
If a shortest path from $u$ to $v$ has length at least four, then
notice that all arcs between $N^+(u)$ and $N^-(v)$ go from $N^-(v)$ to
$N^+(u)$.  Then by items (i) and (ii) from Lemma \ref{lem:endpoints},
we can color $T[N^-(v) \cup N^+(u)]$ with three colors.  Thus, $T$ has a
$(3,1)$-vertex chain, and by Lemma \ref{lem:eff_loc_to_glob}, we can
color $T$ with seven colors.

Next, we consider the case in which a shortest path from $u$ to $v$
has length at most three.  Let $S = N^-(v) \cup N^+(u)$.  By item
(iii) from Lemma \ref{lem:endpoints}, we can color $T[S]$ with five
colors.  Moreover, each remaining vertex is in $N(e)$ for some edge
$e$ on the shortest path, so $T[V\setminus{S}]$ can be colored with
three colors.  So in total, we can color $T$ with at most eight
colors.
\end{proof}

Now it remains to prove Lemma \ref{lem:endpoints}, which we do next.

\begin{proof}[Proof of Lemma \ref{lem:endpoints}]
We will start by establishing some structural claims about light
tournaments which are adapted from \cite{berger2013tournaments}.
Recall that a $C_3$ is a directed triangle.

\begin{definition}\label{def:c3-chain}
Define a {\em $C_3$-chain} of length $\ell$ in $T$ to be a set of
$\ell$ vertex disjoint $C_3$'s, $X = (X_1, X_2, X_3, \ldots,
X_{\ell})$, such that for each $i \in \{1, \ldots, \ell-1\}$, $X_i
\Rightarrow X_{i+1}$.
\end{definition}

A {\em backwards arc} in a $C_3$-chain is an arc $uv$ with $u \in X_i$
and $v \in X_j$ for $j < i$.

\begin{claim}
A $C_3$-chain has no backwards arcs.
\end{claim}

\begin{cproof}
Suppose that there is a backwards arc $e=uv$ with $u \in X_j$ and $v
\in X_i$ for $j > i$ such that $j-i$ is minimum.  Since $X_i
\Rightarrow X_{i+1}$, it must be that $j > i+1$.  So then we have that
$X_{i+1} \subseteq N(e)$, which implies that $e$ is a heavy arc, a
contradiction.\end{cproof}

Let $X = (X_1, X_2, \ldots, X_{\ell})$ be a $C_3$-chain in
$T$, and let $W = V(T)\setminus{V(X)}$.  Initially, $X$ can be of any
length $\ell \geq 1$.

\begin{claim}
For $w \in W$:

\begin{enumerate}

\item If $w \Rightarrow X_i$, then $w \Rightarrow X_j$ for all $j
  \geq i$.

  \item If $X_i \Rightarrow w$, then $X_j \Rightarrow w$ for all $j
  \leq i$.

\end{enumerate}

\end{claim}

\begin{cproof}
  Suppose $w \Rightarrow X_i$ and there is an arc $uw$ with $u \in
  X_j$ for $j > i$.  Then $uw$ is a heavy arc.  Similarly, suppose $X_i
  \Rightarrow w$ and there is an arc $wu$ with $u \in X_j$ for $j <
  i$, then $wu$ is a heavy arc.\end{cproof}

We partition the vertices in $W$ into zones $(Z_0,
Z_1, \ldots, Z_{\ell})$ using the following criteria.  For $w \in W$,
if $i$ is the highest index such that $X_i \Rightarrow w$, then $w$ is
assigned to zone $Z_{i}$.  If there is no such $X_i$, then $w$ is
assigned to zone $Z_0$.

Say a vertex $w \in W$ is {\em clear} if $w \Rightarrow X_i$ or $X_i
\Rightarrow w$ for all $X_i$ in $X$.  Let $C \subseteq W$ be the set
of clear vertices.

\begin{claim}\label{clm:extend}
If $C$ is not transitive, we can extend $X$.
\end{claim}

\begin{cproof}
First, we observe that if $C$ contains a triangle (i.e., $C$ is not
transitive), then $Z_i \cap C$ contains a triangle for some zone
$Z_i$.  This follows from the observation that there are no backwards
arcs from $Z_j \cap C$ to $Z_i \cap C$ for $i < j$.  Indeed, should
such an arc $uv$ from $Z_j \cap C$ to $Z_i \cap C$ exist, then
$X_{i+1} \subset N(uv)$, so $uv$ would be heavy.

If the set $Z_i \cap C$ contains a triangle, then we can extend $X$ by
adding a new triangle to the chain between $X_i$ and $X_{i+1}$.\end{cproof}

We say that $X$ is a {\em maximal $C_3$-chain} if $C$ is transitive.
Let us also now define the {\em unclear} vertices $U$, where $U =
W\setminus{C}$.

\begin{claim}\label{clm:end_triangle}
  If $X = (X_1, \ldots, X_{\ell})$ is a maximal $C_3$-chain, then
for vertex $a \in X_1$, we have 
$N^-(a) \cap U \subseteq \npm(X_1)$.
\end{claim}

\begin{cproof}
If a vertex $u \in N^-(a)$ has $u \Rightarrow X_1$, then $u$ would be
a clear vertex.\end{cproof}

It is not difficult to see that we can efficiently find a maximal
$C_3$-chain; We can begin with any $C_3$-chain and if the associated
set $C$ is not transitive, we can extend the chain by Claim
\ref{clm:extend}.  Let $X = (X_1, \ldots, X_{\ell})$ be a maximal
$C_3$-chain, and let $X_1 = abc$ and let $X_{\ell} = xyz$.  Notice
that $N^-(a) = (N^-(a) \cap C) \cup (N^-(a) \cap U) \cup \{c\}$.  By
Lemma \ref{clm:end_triangle}, we have $N^-(a) \cap U \subseteq
\npm(X_1)$.  Moreover, $\npm(X_1) \subset N(ab) \cup N(bc) \cup N(ca)$.  
 Notice that $c \in N(ab)$ and that for $v \in N^-(a) \cap
U$, $v \notin N(ca)$.  Thus, $(N^-(a) \cap U) \cup \{c\} \subseteq
N(ab) \cup N(bc)$, which is efficiently 2-colorable.

Making an analogous argument for $z \in X_{\ell}$, we conclude that
$((N^+(z) \cup N^-(a)) \cap U) \cup \{y\} \cup \{c\}$ is efficiently
4-colorable.  $T[(N^-(a) \cup N^+(z))\cap C]$ is transitive (since $X$
is maximal) and can be colored with one color.  Therefore, $T[N^+(z)
  \cup N^-(a)]$ can be colored with five colors.  Moreover,
$T[N^+(z)]$ and $T[N^-(a)]$ can each be colored with three
colors.\end{proof}

Further
discussion of light tournaments is included in the next section, in
which we mention some open problems.  We also note that the approach in this section used to bound the chromatic number of light tournaments can be applied to a more general subclass of tournaments.
Specifically, define the hero $H_k$ as follows.
\begin{definition}
Let $\{H_k\}_{0 \leq k}$ be the family of tournaments defined
recursively with $H_0$ being a single vertex and $H_{k+1} =
\Delta(H_k,1,1)$ for all $k \geq 0$.
\end{definition}
Notice that $H_1 = C_3$ and $H_2 = \Delta(C_3,1,1)$.  Our proof that
$H_2$-free tournaments have bounded chromatic number can be extended
in a straightforward way to show that the chromatic number of
$H_k$-free tournaments is upper bounded by a function of $k$.  We omit the details, which can be found in \cite{FelixThesis}.

\section{Conclusion}\label{sec:conclusion}

There are many open questions related to the theorems we have
presented since all the rows in Table \ref{tab:gen_tourn} present gaps
between the upper and lower bounds.  One example is light tournaments,
for which Theorem \ref{thm:bound_light} gives an upper bound on the
chromatic number.  It is a natural question to determine tight upper
and lower bounds on the chromatic number of light tournaments (e.g.,
see Problem 1 in \cite{west-webpage}).  With respect to lower bounds,
there exist light tournaments that are not 2-colorable.  An example of
such a tournament is the Paley tournament $P_7$, one of the four
3-chromatic tournaments on seven vertices~\cite{neumann19943}.  This
tournament is represented in Figure \ref{fig:P7}.  We have not found
any light tournament with chromatic number at least four.  The Paley
tournament $P_{11}$ is the unique 4-chromatic tournament on 11
vertices~\cite{neumann19943}.  A light 4-chromatic tournament would
have to have at least 13 vertices as \cite{bellitto2022smallest}
proved that any $4$-chromatic tournament on $12$ vertices must contain
an induced copy of $P_{11}$, and $P_{11}$ is not light.

\begin{figure}
	\center
	\begin{tikzpicture}
	\tikzstyle{smallvertex}=[circle,draw,minimum size=10pt,inner sep=0pt]
		\node[smallvertex] (v1) at (0,0) {$1$};
		\node[smallvertex] (v2) at (0,1) {$2$};
		\node[smallvertex] (v3) at (2,0) {$3$};
		\node[smallvertex] (v4) at (2,1) {$4$};
		\node[smallvertex] (v5) at (1,2) {$5$};
		\node[smallvertex] (v6) at (-1.5,0.5) {$6$};
		\node[smallvertex] (v7) at (3.5,0.5) {$7$};
		  \draw[->,>=stealth] (v1) -- (v2);
		  \draw[->,>=stealth] (v1) -- (v3);
		  \draw[->,>=stealth] (v1) -- (v4);
		  \draw[->,>=stealth] (v2) -- (v3);
		  \draw[->,>=stealth] (v2) -- (v6);
		  \draw[->,>=stealth] (v2) -- (v7);
		  \draw[->,>=stealth] (v3) -- (v4);
		  \draw[->,>=stealth] (v3) -- (v5);
		  \draw[->,>=stealth] (v3) -- (v6);
		  \draw[->,>=stealth] (v4) -- (v5);
		  \draw[->,>=stealth] (v4) -- (v2);
		  \draw[->,>=stealth] (v4) -- (v7);
		  \draw[->,>=stealth] (v5) -- (v6);
		  \draw[->,>=stealth] (v5) -- (v1);
		  \draw[->,>=stealth] (v5) -- (v2);
		  \draw[->,>=stealth] (v6) -- (v1);
		  \draw[->,>=stealth] (v6) -- (v4);
		  \draw[->,>=stealth] (v6) -- (v7);
		  \draw[->,>=stealth] (v7) -- (v1);
		  \draw[->,>=stealth] (v7) -- (v3);
		  \draw[->,>=stealth] (v7) -- (v5);
        \end{tikzpicture}
	\caption{A 3-chromatic light tournament. \label{fig:P7}}
\end{figure}
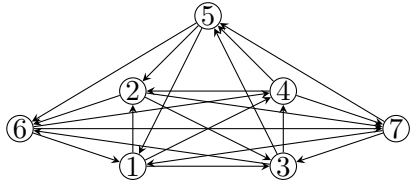

\begin{problem}
  What is the maximum chromatic number of a light tournament?
  \end{problem}

Regarding the complexity of coloring a light tournament, notice that
if it is \NP-hard to color a 2-colorable tournament with four colors
(rather than three as per Theorem \ref{thm:2-hardness}), this would
imply (by Observation \ref{obs:light}) \NP-hardness of coloring a 2-colorable light tournament with two
colors.  Indeed, we have no hardness
results for coloring light tournaments.

\begin{problem}
Is there a polynomial-time algorithm to color a 2-colorable light
tournament with two colors?  Or is this problem \NP-hard?
  \end{problem}

Another interesting topic is the relation of coloring tournaments and the
feedback vertex set (FVS) problem on tournaments.  There is an elegant
2-approximation for this problem~\cite{lokshtanov20212}.  Notice that
Theorem \ref{thm:2-col} implies that in a 2-colorable tournament, we
can efficiently find a FVS of size at most $9n/10$.  In contrast, the
algorithm in \cite{lokshtanov20212} could just return the whole vertex
set if the two transitive sets were of roughly equal size.  
The next problem is analogous to a well-studied question for general
graphs~\cite{dinur2010hardness,khot2014hardness}.

\begin{problem}
  What is
the largest transitive induced subtournament that one can efficiently find in a
2-colorable tournament?  Is it larger than $n/10$?
\end{problem}

Finally, we remark that an implication of Theorem
\ref{thm:3-col_reduction} is that proving any hardness of coloring
3-colorable tournaments would then provide hardness of coloring
3-colorable graphs with 50 times fewer colors.  Since it has taken
around 20 years to go from proving \NP-hardness of coloring a
3-colorable graph with four
colors~\cite{khanna2000hardness,guruswami2000hardness,guruswami2004hardness}
to \NP-hardness of coloring a 3-colorable graph with five
colors~\cite{bulin2019algebraic}, it would be interesting to see if we
can prove hardness of coloring 3-colorable tournaments for a constant
larger than five (at least five is shown in Theorem
\ref{thm:hardness}), or perhaps show that the two problems are
actually equivalent.  In fact, Theorem \ref{thm:graph_to_tourn}
already shows one direction.  The other direction remains
open.

\begin{problem}
Suppose we can color an 3-colorable graph with $\ell > 3$ colors.  Then can we
color a 3-colorable tournament with $\ell$ colors?
\end{problem}

\section*{Acknowledgements} We thank Louis Esperet for useful
discussions and for his encouragement.  We thank the anonymous
referees for their thorough and insightful comments.

\bibliography{hero}

\newcommand{\etalchar}[1]{$^{#1}$}
\begin{thebibliography}{LMM{\etalchar{+}}21}

\bibitem[AKMR96]{alon1996coloring}
Noga Alon, Pierre Kelsen, Sanjeev Mahajan, and Hariharan Ramesh.
\newblock Coloring 2-colorable hypergraphs with a sublinear number of colors.
\newblock {\em Nordic Journal of Computing}, 3:425--439, 1996.

\bibitem[APS01]{alon2001ramsey}
Noga Alon, J{\'a}nos Pach, and J{\'o}zsef Solymosi.
\newblock Ramsey-type theorems with forbidden subgraphs.
\newblock {\em Combinatorica}, 21(2):155--170, 2001.

\bibitem[BBKP24]{bellitto2022smallest}
Thomas Bellitto, Nicolas Bousquet, Adam Kabela, and Th{\'e}o Pierron.
\newblock The smallest 5-chromatic tournament.
\newblock {\em Mathematics of Computation}, 93(345):443--458, 2024.

\bibitem[BCC{\etalchar{+}}13]{berger2013tournaments}
Eli Berger, Krzysztof Choromanski, Maria Chudnovsky, Jacob Fox, Martin Loebl,
  Alex Scott, Paul Seymour, and St{\'e}phan Thomass{\'e}.
\newblock Tournaments and colouring.
\newblock {\em Journal of Combinatorial Theory, Series B}, 103(1):1--20, 2013.

\bibitem[BKO19]{bulin2019algebraic}
Jakub Bul{\'\i}n, Andrei Krokhin, and Jakub Opr{\v{s}}al.
\newblock Algebraic approach to promise constraint satisfaction.
\newblock In {\em Proceedings of the 51st Annual ACM Symposium on Theory of
  Computing (STOC)}, pages 602--613, 2019.

\bibitem[Blu94]{blum1994new}
Avrim Blum.
\newblock New approximation algorithms for graph coloring.
\newblock {\em Journal of the ACM}, 41(3):470--516, 1994.

\bibitem[BRSW12]{barak20122}
Boaz Barak, Anup Rao, Ronen Shaltiel, and Avi Wigderson.
\newblock 2-{S}ource dispersers for $n^{o(1)}$ entropy, and {R}amsey graphs
  beating the {F}rankl-{W}ilson construction.
\newblock {\em Annals of Mathematics}, pages 1483--1543, 2012.

\bibitem[CF96]{chen2005coloring}
Hui Chen and Alan Frieze.
\newblock Coloring bipartite hypergraphs.
\newblock In {\em Fifth International Conference on Integer Programming and
  Combinatorial Optimization (IPCO)}, pages 345--358, 1996.

\bibitem[Chu14]{chudnovsky2014erdos}
Maria Chudnovsky.
\newblock The {E}rd{\H{o}}s-{H}ajnal {C}onjecture—{A} survey.
\newblock {\em Journal of Graph Theory}, 75(2):178--190, 2014.

\bibitem[CHZ07]{chen2007min}
Xujin Chen, Xiaodong Hu, and Wenan Zang.
\newblock A min-max theorem on tournaments.
\newblock {\em SIAM Journal on Computing}, 37(3):923--937, 2007.

\bibitem[CSSS24]{chudnovsky2024pure}
Maria Chudnovsky, Alex Scott, Paul Seymour, and Sophie Spirkl.
\newblock Pure pairs. {X}. {T}ournaments and the strong {E}rd{\H{o}}s-{H}ajnal
  property.
\newblock {\em European Journal of Combinatorics}, 115:103786, 2024.

\bibitem[DKPS10]{dinur2010hardness}
Irit Dinur, Subhash Khot, Will Perkins, and Muli Safra.
\newblock Hardness of finding independent sets in almost 3-colorable graphs.
\newblock In {\em Proceedings of the 51st Annual IEEE Symposium on Foundations
  of Computer Science (FOCS)}, pages 212--221, 2010.

\bibitem[DMR09]{dinur2009conditional}
Irit Dinur, Elchanan Mossel, and Oded Regev.
\newblock Conditional hardness for approximate coloring.
\newblock {\em SIAM Journal on Computing}, 39(3):843--873, 2009.

\bibitem[DRS05]{dinur2005hardness}
Irit Dinur, Oded Regev, and Clifford Smyth.
\newblock The hardness of 3-uniform hypergraph coloring.
\newblock {\em Combinatorica}, 25(5):519--535, 2005.

\bibitem[EE85]{el1985existence}
Mohamed {El-Zahar} and Paul Erd{\H{o}}s.
\newblock On the existence of two non-neighboring subgraphs in a graph.
\newblock {\em Combinatorica}, 5:295--300, 1985.

\bibitem[EH89]{erdos1989ramsey}
Paul Erd\H{o}s and Andr{\'a}s Hajnal.
\newblock Ramsey-type theorems.
\newblock {\em Discrete Applied Mathematics}, 25(1-2):37--52, 1989.

\bibitem[EM64]{erdos1964representation}
Paul Erdos and Leo Moser.
\newblock On the representation of directed graphs as unions of orderings.
\newblock {\em Math. Inst. Hung. Acad. Sci}, 9:125--132, 1964.

\bibitem[FGSY19]{fox2019removal}
Jacob Fox, Lior Gishboliner, Asaf Shapira, and Raphael Yuster.
\newblock The removal lemma for tournaments.
\newblock {\em Journal of Combinatorial Theory, Series B}, 136:110--134, 2019.

\bibitem[FHS19]{feder2019complexity}
Tom\'as Feder, Pavol Hell, and Carlos Subi.
\newblock Complexity of acyclic colorings of graphs and digraphs with degree
  and girth constraints.
\newblock {\em arXiv:1907.00061}, 2019.

\bibitem[FK98]{feige1998zero}
Uriel Feige and Joe Kilian.
\newblock Zero knowledge and the chromatic number.
\newblock {\em Journal of Computer and System Sciences}, 57(2):187--199, 1998.

\bibitem[GK00]{guruswami2000hardness}
Venkatesan Guruswami and Sanjeev Khanna.
\newblock On the hardness of 4-coloring a 3-colorable graph.
\newblock In {\em Proceedings 15th Annual IEEE Conference on Computational
  Complexity (CCC)}, pages 188--197, 2000.

\bibitem[GK04]{guruswami2004hardness}
Venkatesan Guruswami and Sanjeev Khanna.
\newblock On the hardness of 4-coloring a 3-colorable graph.
\newblock {\em SIAM Journal on Discrete Mathematics}, 18(1):30--40, 2004.

\bibitem[GS20]{guruswami2020d}
Venkatesan Guruswami and Sai Sandeep.
\newblock $d$-{T}o-1 hardness of coloring 3-colorable graphs with {$O(1)$}
  colors.
\newblock In {\em 47th International Colloquium on Automata, Languages, and
  Programming (ICALP)}, 2020.

\bibitem[Hal93]{halldorsson1993still}
Magn{\'u}s~M. Halld{\'o}rsson.
\newblock A still better performance guarantee for approximate graph coloring.
\newblock {\em Information Processing Letters}, 45(1):19--23, 1993.

\bibitem[Has99]{hastadClique}
Johan Hastad.
\newblock Clique is hard to approximate within $n^{1-\epsilon}$.
\newblock {\em Acta Mathematica}, 182:105--142, 1999.

\bibitem[HLNT19]{harutyunyan2019coloring}
Ararat Harutyunyan, Tien-Nam Le, Alantha Newman, and St{\'e}phan Thomass{\'e}.
\newblock Coloring dense digraphs.
\newblock {\em Combinatorica}, 39(5):1021--1053, 2019.

\bibitem[HLTW19]{harutyunyan2019locToGlobal}
Ararat Harutyunyan, Tien-Nam Le, St{\'e}phan Thomass{\'e}, and Hehui Wu.
\newblock Coloring tournaments: {F}rom local to global.
\newblock {\em Journal of Combinatorial Theory, Series B}, 138:166--171, 2019.

\bibitem[Kho02]{khot2002power}
Subhash Khot.
\newblock On the power of unique 2-prover 1-round games.
\newblock In {\em Proceedings of the 34th Annual ACM Symposium on Theory of
  Computing (STOC)}, pages 767--775, 2002.

\bibitem[Kli23]{FelixThesis}
Felix Klingelhoefer.
\newblock {\em Algorithms for Promise Coloring Problems on Tournaments and
  Graphs}.
\newblock PhD thesis, Universit\'e Grenoble Alpes, 2023.

\bibitem[KLS00]{khanna2000hardness}
Sanjeev Khanna, Nathan Linial, and Shmuel Safra.
\newblock On the hardness of approximating the chromatic number.
\newblock {\em Combinatorica}, 20(3):393--415, 2000.

\bibitem[KMS98]{KargerMS98}
David~R. Karger, Rajeev Motwani, and Madhu Sudan.
\newblock Approximate graph coloring by semidefinite programming.
\newblock {\em Journal of the ACM}, 45(2):246--265, 1998.

\bibitem[KN24]{klingelhoefer2023bounding}
Felix Klingelhoefer and Alantha Newman.
\newblock Bounding the chromatic number of dense digraphs by arc neighborhoods.
\newblock {\em Combinatorica}, 44(4):881--895, 2024.

\bibitem[KNS01]{krivelevich2001approximating}
Michael Krivelevich, Ram Nathaniel, and Benny Sudakov.
\newblock Approximating coloring and maximum independent sets in 3-uniform
  hypergraphs.
\newblock {\em Journal of Algorithms}, 41(1):99--113, 2001.

\bibitem[KS14]{khot2014hardness}
Subhash Khot and Rishi Saket.
\newblock Hardness of finding independent sets in 2-colorable and almost
  2-colorable hypergraphs.
\newblock In {\em Proceedings of the 25th Annual ACM-SIAM Symposium on Discrete
  Algorithms (SODA)}, pages 1607--1625, 2014.

\bibitem[KT17]{kawarabayashi2017coloring}
Ken-ichi Kawarabayashi and Mikkel Thorup.
\newblock Coloring 3-colorable graphs with less than $n^{1/5}$ colors.
\newblock {\em Journal of the ACM}, 64(1):1--23, 2017.

\bibitem[LMM{\etalchar{+}}21]{lokshtanov20212}
Daniel Lokshtanov, Pranabendu Misra, Joydeep Mukherjee, Fahad Panolan,
  Geevarghese Philip, and Saket Saurabh.
\newblock 2-{A}pproximating feedback vertex set in tournaments.
\newblock {\em ACM Transactions on Algorithms}, 17(2):1--14, 2021.

\bibitem[Lov73]{lovasz1973coverings}
L{\'a}szl{\'o} Lov{\'a}sz.
\newblock Coverings and colorings of hypergraphs.
\newblock In {\em Proc. 4th Southeastern Conference of Combinatorics, Graph
  Theory, and Computing}, pages 3--12, 1973.

\bibitem[MW11]{west-webpage}
Kevin Milans and Douglas~B. West.
\newblock Chromatic numbers of tournaments, 2011.
\newblock \url{https://dwest.web.illinois.edu/regs/chromtourn.html}, Last
  accessed on 2023-09-11.

\bibitem[{Neu}82]{neumann1982dichromatic}
Victor {Neumann-Lara}.
\newblock The dichromatic number of a digraph.
\newblock {\em Journal of Combinatorial Theory, Series B}, 33(3):265--270,
  1982.

\bibitem[NL94]{neumann19943}
Victor Neumann-Lara.
\newblock The 3 and 4-dichromatic tournaments of minimum order.
\newblock {\em Discrete Mathematics}, 135(1-3):233--243, 1994.

\bibitem[NSS23]{nguyen2023some}
Tung Nguyen, Alex Scott, and Paul Seymour.
\newblock Some results and problems on tournament structure.
\newblock {\em preprint arXiv:2306.02364}, 2023.

\bibitem[Wig83]{wigderson1983improving}
Avi Wigderson.
\newblock Improving the performance guarantee for approximate graph coloring.
\newblock {\em Journal of the ACM}, 30(4):729--735, 1983.

\bibitem[Zuc06]{zuckerman2006linear}
David Zuckerman.
\newblock Linear degree extractors and the inapproximability of max clique and
  chromatic number.
\newblock In {\em Proceedings of the 38th Annual ACM Symposium on Theory of
  Computing (STOC)}, pages 681--690, 2006.

\end{thebibliography}

\appendix

\section{\NP-Hardness of Deciding 2-Colorability}\label{sec:base_hardness}

For completeness, we provide a proof of the \NP-hardness of coloring
$2$-colorable tournaments with two colors. This proof is strongly
inspired by the proof of \cite{chen2007min}.  Notice that Lemma \ref{lem:np-hard2} immediately implies that it is \NP-hard to color a 2-colorable tournament with two colors.  

\begin{lemma}\label{lem:np-hard2}
  It is \NP-hard to decide if a tournament is 2-colorable.
  \end{lemma}

\begin{proof}
We will reduce this problem to the problem of deciding 2-colorability
of 3-uniform hypergraphs, which is known to be \NP-hard
\cite{dinur2005hardness}.  Let $\Hd = (\V, \Ed)$ be a 3-uniform
hypergraph.  We now build a tournament $T=(V,A)$ such that $T$ is
2-colorable iff $\Hd$ is 2-colorable.

We will start by defining a subtournament $T_1=(V_1,A_1)$ of $T$.  Fix
an arbitrary ordering of the hyperedges of $\Hd$, namely ${\Ed} =
\{e_1, e_2, \ldots, e_m\}$.  For each $e_i =(v_a,v_b,v_c)$ in $\Ed$,
we add three vertices $v_{a,i}$, $v_{b,i}$ and $v_{c,i}$ to $V_1$, and
add to $A_1$ the arcs $(v_{a,i},v_{b,i})$, $(v_{b,i},v_{c,i})$ and
$(v_{c,i},v_{a,i})$ such that these three vertices form a directed
triangle. We then add the arcs from all the vertices $v_{a,i}$ towards
all the vertices $v_{b,j}$ for any $a,b,i,j$ with $i<j$.  We now
define a new subtournament $T_2=(V_2,A_2)$ made up of three vertices
that form a directed triangle.  Finally, we define a last
subtournament $T_3=(V_3,A_3)$, where $V_3 := \V$, and $T_3$ forms a
transitive tournament on $V_3$.

Then add $T_1$, $T_2$ and $T_3$ to $T$. Orient all arcs from vertices
in $V_1$ towards vertices of $V_2$ and all arcs from vertices of $V_2$
towards vertices of $V_3$.  The only arcs we still need to orient are
those between $V_1$ and $V_3$.  For this, we look at the vertices of
$\Hd$ from which the vertices of $V_1$ and $V_3$ are derived; for
$v_{a,i} \in V_1$ and $v'_{b} \in V_3$, we add an arc from $v'_{b}$ to
$v_{a,i}$ iff $a=b$ (i.e., if they are derived from the same vertex of
$\Hd$), and we add an arc from $v_{a,i}$ to $v'_{b}$ otherwise. This
completes the definition of $T$. Figure \ref{fig:basic_hardness} gives
an example of this construction for a hypergraph with five vertices
and four hyperedges.

We will now establish that if $\Hd$ is $2$-colorable, so is $T$. Given
a $2$-coloring of $\Hd$, give all the vertices of $V_1$ the same color
as the vertex of $\Hd$ they are derived from, and those in $V_3$ the
opposite color of the vertex of $\Hd$ they are derived from. Finally
color $T_2$ properly with the same two colors.  Any arc that goes from
$V_3$ to $V_1$ will be bicolored, and since all arcs are oriented from
$V_1$ towards $V_2$ and from $V_2$ towards $V_3$, there can only be
monochromatic triangles inside $V_1$, $V_2$ or $V_3$. However, $V_2$
is a bicolored triangle, and every triangle in $V_1$ represents a
hyperedge of $\Hd$ and must therefore contain two vertices of
different colors.  Furthermore, the set $V_3$ contains no triangles.

It remains to show that if $\Hd$ has chromatic number at least $3$,
$T$ has chromatic number at least $3$.  We will establish this by
contradiction: We show that if $T$ has a proper 2-coloring $C$, then
we can construct a proper 2-coloring of $\Hd$.

We define a coloring $C_H$ of $\Hd$ by assigning to every vertex $v_a
\in \V$ the same color as its corresponding vertex $v_a' \in V_3$ has
in $C$.  Let us show that $C_H$ is a proper 2-coloring of $\Hd$.
Notice that in a proper 2-coloring of $T$, $v_{a,i} \in V_1$ must have
the opposite color of $v_a' \in V_3$, for any $a,i$. If this were not
the case, these two vertices would form a directed triangle with a
vertex in $V_2$ of the same color, which exists since $T_2$ is a
directed triangle and must be bicolored.  Now suppose some hyperedge
$e_i=(v_a,v_b,v_c)$ is monochromatic under $C_H$. Then $v_a', v_b',
v_c' \in V_3$ all have the same color.  Then, there is a triangle
$(v_{a,i},v_{b,i},v_{c,i})$ in $T_1$ by definition, and all its
vertices must have the same color (the opposite of that used for
$v_a', v_b', v_c'$).  This is a contradiction, and therefore all
hyperedges of $\Hd$ are bicolored using the coloring $C_H$.
\end{proof}

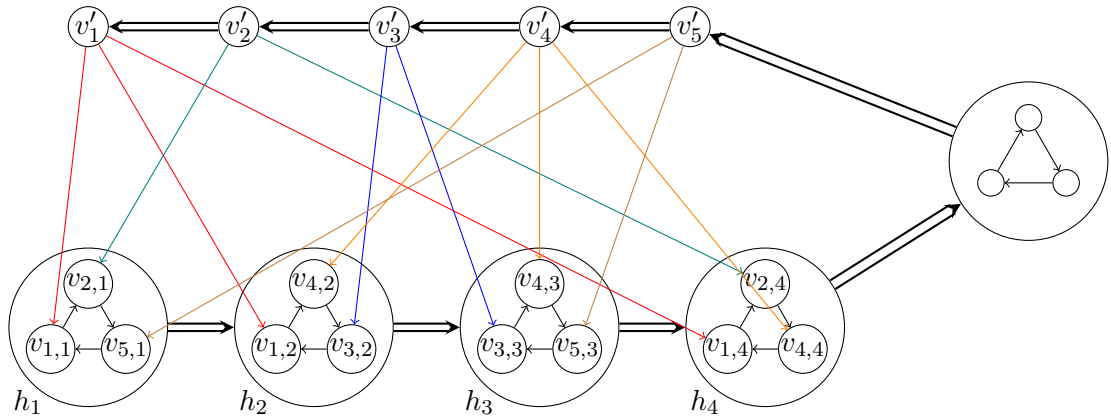
\begin{figure}
\centering
\begin{tikzpicture}
	  \tikzstyle{smallvertex}=[circle,draw,minimum size=10pt,inner sep=0pt]
	  \tikzstyle{bigvertex}=[circle,draw,minimum size=40pt,inner sep=0pt]
	  \tikzstyle{group}==[circle,draw,minimum size=60pt,inner sep=0pt]

	  \node[smallvertex] (v1) at (0,0) {$v_{1,1}$};
	  \node[smallvertex] (v2) at (0.5,0.866) {$v_{2,1}$};
	  \node[smallvertex] (v3) at (1,0) {$v_{5,1}$};
	  \node[group] (g1) at (0.5,0.289) {$$};
	  \node at (-0.3,-0.7) {$e_1$};
	  \draw [->] (v1) to (v2);
	  \draw [->] (v2) to (v3);
	  \draw [->] (v3) to (v1);

	  \begin{scope}[xshift=3cm]
	  \node[smallvertex] (v4) at (0,0) {$v_{1,2}$};
          \node[smallvertex] (v5) at (0.5,0.866) {$v_{4,2}$};
          \node[smallvertex] (v6) at (1,0) {$v_{3,2}$};
          \node[group] (g2) at (0.5,0.289) {$$};
	  \node at (-0.3,-0.7) {$e_2$};
          \draw [->] (v4) to (v5);
          \draw [->] (v5) to (v6);
          \draw [->] (v6) to (v4);
	  \end{scope}

          \begin{scope}[xshift=6cm]
          \node[smallvertex] (v7) at (0,0) {$v_{3,3}$};
          \node[smallvertex] (v8) at (0.5,0.866) {$v_{4,3}$};
          \node[smallvertex] (v9) at (1,0) {$v_{5,3}$};
          \node[group] (g3) at (0.5,0.289) {$$};
	  \node at (-0.3,-0.7) {$e_3$};
          \draw [->] (v7) to (v8);
          \draw [->] (v8) to (v9);
          \draw [->] (v9) to (v7);
          \end{scope}

          \begin{scope}[xshift=9cm]
          \node[smallvertex] (v10) at (0,0) {$v_{1,4}$};
          \node[smallvertex] (v11) at (0.5,0.866) {$v_{2,4}$};
          \node[smallvertex] (v12) at (1,0) {$v_{4,4}$};
          \node[group] (g4) at (0.5,0.289) {$$};
	  \node at (-0.3,-0.7) {$e_4$};
          \draw [->] (v10) to (v11);
          \draw [->] (v11) to (v12);
          \draw [->] (v12) to (v10);
          \end{scope}

	  \draw [->,>=stealth,thick,double distance=2pt] (g1) to (g2);
	  \draw [->,>=stealth,thick,double distance=2pt] (g2) to (g3);
	  \draw [->,>=stealth,thick,double distance=2pt] (g3) to (g4);

	  \begin{scope}[yshift=4cm]
          \node[smallvertex] (w1) at (0.5,0.289) {$v_1'$};
          \node[smallvertex] (w2) at (2.5,0.289) {$v_2'$};
          \node[smallvertex] (w3) at (4.5,0.289) {$v_3'$};
          \node[smallvertex] (w4) at (6.5,0.289) {$v_4'$};
          \node[smallvertex] (w5) at (8.5,0.289) {$v_5'$};

	  \draw [->,>=stealth,thick,double distance=2pt] (w2) to (w1);
	  \draw [->,>=stealth,thick,double distance=2pt] (w3) to (w2);
	  \draw [->,>=stealth,thick,double distance=2pt] (w4) to (w3);
	  \draw [->,>=stealth,thick,double distance=2pt] (w5) to (w4);
	  \end{scope}

	  \node[group] (G) at (13,2.5) {$$};
	  \node[smallvertex] (t1) at (12.5,2.211) {$$};
	  \node[smallvertex] (t2) at (13,3.077) {$$};
	  \node[smallvertex] (t3) at (13.5,2.211) {$$};
          \draw [->] (t1) to (t2);
          \draw [->] (t2) to (t3);
          \draw [->] (t3) to (t1);

          \draw [->,>=stealth,thick,double distance=3pt] (g4) to (G);
          \draw [->,>=stealth,thick,double distance=3pt] (G) to (w5);
          \draw [->,red] (w1) to (v1);
          \draw [->,red] (w1) to (v4);
          \draw [->,red] (w1) to (v10);
          \draw [->,teal] (w2) to (v2);
          \draw [->,teal] (w2) to (v11);
          \draw [->,blue] (w3) to (v6);
          \draw [->,blue] (w3) to (v7);
          \draw [->,orange] (w4) to (v5);
          \draw [->,orange] (w4) to (v8);
          \draw [->,orange] (w4) to (v12);
          \draw [->,brown] (w5) to (v3);
          \draw [->,brown] (w5) to (v9);

\end{tikzpicture}
\caption{Construction of $T$ from a 3-uniform hypergraph $\Hd$.  There
  is a downwards arc between $v_b'$ and all vertices $v_{b,i}$ for
  every $b,i$. These are the colored arcs in the figure.  All
  remaining arcs all go upwards from the vertices $v_{a,i}$ towards the
  vertices $v_{b}'$ for $a \neq b$.
\label{fig:basic_hardness}}
\end{figure}

\section{Hardness of Approximation for General Tournaments}\label{sec:hardness_gen_app}

In this section, we prove Theorem \ref{thm:gen_hardness}.  Our proof
parallels the proof of hardness of approximate coloring of digon-free
digraphs of \cite{feder2019complexity}; we extend their approach to
tournaments and show that it can be used to obtain hardness of
approximation.

\genhardness*

The proof uses an explicit construction for bipartite Ramsey graphs.
A bipartite graph $G=(X\cup Y, E)$ is $k$-Ramsey if for every $S_X
\subset X$ and $S_Y \subset Y$ where $|S_X| = |S_Y| = k$, the subgraph
induced on the vertex set $S_X \cup S_Y$ is neither empty nor a
complete bipartite subgraph.  Specifically, we use the main theorem
from \cite{barak20122}, which we include here (slightly rephrased) for
completeness.

\begin{theorem}[Theorem 1.3 in \cite{barak20122}]\label{thm:barak}
For every large enough integer $n$, there is an explicit construction
of a bipartite $n^{o(1)}$-Ramsey graph on $2n$ vertices.
  \end{theorem}

\begin{lemma}\label{lem:ramsey}
Let $\eps$ be a constant such that $0 < \eps <1$ and let $n$ be a
sufficiently large integer. Then there exists a tournament $T=(V,A)$
where $V=X\cup Y$ and $|X| = |Y| = n$, such that for every two subsets
$S_X \subseteq X$, $S_Y \subseteq Y$ with $|S_X| \geq n^\eps$ and
$|S_Y| \geq n^\eps$, the tournament induced on $S_X \cup S_Y$ contains
a triangle.
\end{lemma}

\begin{proof}
By Theorem \ref{thm:barak}, for sufficiently large $n$, there exists
an explicit construction of a bipartite $n^{o(1)}$-Ramsey graph over
$n$ vertices.  Let $B_1=(X_1,Y_1,E_1)$ be such a graph.  Define
the tournament $T=(V,E)$ with $V=X\cup Y$ as follows:
\begin{itemize}
\item $X = X_1$ and $Y = Y_1$.
\item Orient the arcs inside $X$ and $Y$ such that they both induce
                        transitive tournaments.
\item For every $u \in X$, $v \in Y$, orient the arc from $u$ to $v$ if
                        $(u,v)\in E_1$ and from $v$ to $u$ otherwise.
\end{itemize}

Given $0 < \eps <1$, take any $S_X \subseteq X$, $S_Y \subseteq Y$
with $|S_X| \geq n^\eps$ and $|S_Y| \geq n^\eps$.  Let $x \in S_X$ and
$y \in S_Y$ be middle vertices of $S_X$ and $S_Y$ (i.e., $x$ has
roughly equal in and out-degree in $S_X$, and $y$ in $S_Y$). Without
loss of generality, suppose that the arc between $x$ and $y$ is
oriented from $x$ to $y$. Then the subgraph of $B_1$ induced on the
vertex set $S_X \cap N^-(x)$ and $S_Y \cap N^+(y)$ has at least
$n^{\eps}/2 -1$ vertices on each side of the bipartition.  Thus, it is
neither complete nor empty (since for sufficiently large $n$,
$n^{o(1)} \leq n^{\eps}/2-1$). This implies that there is an arc from
a vertex $v \in S_Y \cap N^+(y)$ to a vertex $u \in S_X\cap N^-(x)$.
So there is a directed cycle on the four vertices $v,u,x,y$ in $T[S_X
  \cup S_Y]$.  Since, $T$ is a tournament, there is also a directed
triangle using three of these four vertices.
\end{proof}

\begin{theorem}\label{thm:hard-acyclic}
It is $\NP$-hard (under randomized reductions) to decide if a
tournament has an acyclic induced subgraph of size at most 
$n^{1/2+\varepsilon}$ or can be colored with at most $n^\eps$ colors, for
every $0<\varepsilon<\frac{1}{2}$.
\end{theorem}

\begin{proof}
For any $\eps > 0$, let $G=(V,E)$ be a graph on $n_G$ vertices.  Feige
and Kilian~\cite{feige1998zero} proved that it is \NP-hard (under
randomized reductions) to decide if $\alpha(G) < n_G^{\eps}$ or if
$\chi(G) \leq n_G^{\eps}$ for every $\eps > 0$.  (As is standard,
$\alpha(G)$ denotes the size of the maximum independent set in $G$ and
$\chi(G)$ denotes its chromatic number.)

For each vertex $v_i \in V$, define a new transitive tournament $T_i$
on $n_G$ vertices. For each edge $(v_i,v_j) \in E$ with $i < j$, join
$T_i$ and $T_j$ such that they form the tournament of Lemma
\ref{lem:ramsey}, with $T_i$ being $X$ and $T_j$ being $Y$. For all
remaining $v_i,v_j \in V$ with $i<j$ (such that $(v_i,v_j) \notin E$),
orient all arcs from each vertex of $T_i$ to each vertex of
$T_j$. This defines a new tournament $T$ on $n=n_G^2$ vertices.  

Suppose that $\chi(G) \leq n_G^{\eps}$.  By coloring
vertices in $T_i$ with the color of $v_i$ in $G$, we see that $T$ has
a coloring with at most $n_G^{\eps} = n^{\eps/2}$ colors.  Indeed, the
only arcs in $T$ that are not bicolored are inside a $T_i$ for some
$i$, or from a vertex of $T_i$ to a vertex of $T_j$ for $i<j$, and can
thus never form a triangle.

Now consider the other case, in which $\alpha(G) < n_G^{\eps}$.
Let $S$ be an induced acyclic subtournament of $T$.  By Lemma
\ref{lem:ramsey}, if $S$ intersects some $T_i$'s on more than
$n_G^{\eps}$ vertices each, then the respective $v_i$'s form an
independent set in $G$.  Therefore, if $|S| \geq n^{1/2+\eps/2}$, there
must be at least $n_G^\eps$ tournaments $T_i$ that intersect $S$ on at
least $n_G^{\eps}$ vertices, which then leads to an independent set of
size at least $n_G^{\eps}$ in $G$.  So we conclude that $|S| < n^{1/2
  + \eps/2}$. 
\end{proof}

The next corollary follows from the fact that if the maximum size of
an induced acyclic subgraph in a tournament is at most $n^{1/2 +
  \varepsilon}$, then its dichromatic number is at least $n^{1/2-\varepsilon}$.

\begin{corollary}\label{cor:decide}
It is \NP-hard (under randomized reductions) to decide if a tournament
has dichromatic number at least $n^{1/2-\varepsilon}$ or at most
$n^{\varepsilon}$ for every $0 < \varepsilon < \frac{1}{2}$.  
\end{corollary}

We are now ready to prove Theorem \ref{thm:gen_hardness}.  Suppose
that we have an algorithm to approximate the chromatic number of a
tournament to within a factor of $n^{1/2 - 3\varepsilon}$.   Then we
could distinguish between the two cases stated in Corollary
\ref{cor:decide}.
Setting $\delta = 3\varepsilon$, we arrive at the statement of Theorem
\ref{thm:gen_hardness}.

\end{document}